\documentclass[12pt]{article}
\usepackage[utf8]{inputenc}
\usepackage[dvipsnames]{xcolor}
\usepackage{amsthm,amsmath,amssymb,bbm,hyperref,graphicx,float,tikz,comment}
\usepackage{mathtools} 
\usepackage[normalem]{ulem} 
\usepackage[affil-it]{authblk}
\allowdisplaybreaks

\newtheorem{thm}{Theorem}
\newtheorem{prop}{Proposition}
\newtheorem{lemma}{Lemma}
\newtheorem{cor}{Corollary}
\theoremstyle{definition}
\newtheorem{rmk}{Remark}

\theoremstyle{remark}

\DeclareMathOperator{\tr}{tr}

\DeclareMathOperator{\argmax}{arg\,max}
\newcommand{\RE}{\textup{Re\,}}

\newcommand{\Hilbert}{\mathcal{H}}

\newcommand{\PPP}{\mathbb{P}}

\newcommand{\ZZZ}{\mathbb{Z}}
\newcommand{\SSS}{\mathbb{S}}
\newcommand{\be}{\begin{equation}}
\newcommand{\ee}{\end{equation}}

\newcommand{\scp}[2]{\langle #1,#2 \rangle}

\newcommand{\eps}{\varepsilon}
\newcommand{\mc}{\mathrm{mc}}

\newcommand{\eq}{\mathrm{eq}}
\newcommand{\noneq}{\mathrm{neq}}
\newcommand{\HfF}{H_0^\mathrm{fF}}

\makeatletter
\newcommand{\specialcell}[1]{\ifmeasuring@#1\else\omit$\displaystyle#1$\ignorespaces\fi}
\makeatother

\title{Macroscopic Thermalization for Highly Degenerate Hamiltonians After Slight Perturbation}
\date{July 11, 2025}
\author[1]{Barbara Roos\thanks{ORCID: 0000-0002-9071-5880, E-mail: barbara.roos@uni-tuebingen.de}}
\author[2,3]{Shoki Sugimoto\thanks{ORCID: 0000-0002-5874-928X, E-mail: shoki.sugimoto@ap.t.u-tokyo.ac.jp}}
\author[1]{Stefan Teufel\thanks{ORCID: 0000-0003-3296-4261, E-mail: stefan.teufel@uni-tuebingen.de}}
\author[1]{Roderich Tumulka\thanks{ORCID: 0000-0001-5075-9929, E-mail: roderich.tumulka@uni-tuebingen.de}}
\author[1,4]{Cornelia Vogel\thanks{ORCID: 0000-0002-3905-4730, E-mail: cornelia.vogel@uni-tuebingen.de}}
\affil[1]{Mathematics Institute, Eberhard Karls University T\"ubingen, Auf der Morgenstelle 10, 72076 T\"ubingen, Germany}
\affil[2]{Department of Applied Physics, The University of Tokyo, 7-3-1 Hongo, Bunkyo-ku, Tokyo 113-0033, Japan}
\affil[3]{Nonequilibrium Quantum Statistical Mechanics RIKEN Hakubi Research Team, RIKEN Pioneering Research Institute (PRI), Wako, Saitama 351-0198, Japan}
\affil[4]{Dipartimento di Matematica, Universit\`a degli Studi di Milano, Via Cesare Saldini 50, 20133 Milano, Italy}

\begin{document}

\maketitle

\begin{abstract}
We say of an isolated macroscopic quantum system in a pure state $\psi$ that it is in macroscopic thermal equilibrium (MATE) if $\psi$ lies in or close to a suitable subspace $\Hilbert_\eq$ of Hilbert space. It is known that
every initial state $\psi_0$ will eventually reach and stay there most of the time (``thermalize'') if the Hamiltonian is non-degenerate and satisfies the appropriate version of the eigenstate thermalization hypothesis (ETH), i.e., that every eigenvector is in MATE. Tasaki recently proved the ETH for a certain perturbation $H_\theta^\mathrm{fF}$ of the Hamiltonian $\HfF$ of $N\gg 1$ free fermions on a one-dimensional lattice. 
The perturbation is needed to remove the high degeneracies of $\HfF$. 
Here, we first point out that also for degenerate Hamiltonians all $\psi_0$ thermalize if the ETH holds, i.e., if \emph{every} eigenbasis lies in MATE, and we prove that this is the case for $\HfF$. Inspired by the fact that there is \emph{one} eigenbasis of $\HfF$ for which MATE can be proved more easily than for the others,  with smaller error bounds, and also in higher spatial dimensions, we show for any given $H_0$ that the existence of  one eigenbasis in MATE implies quite generally that \emph{most} eigenbases of $H_0$ lie in MATE. 
We also show that, as a consequence,  after adding a small generic perturbation, $H=H_0+\lambda V$ with $\lambda\ll 1$, for most perturbations~$V$ the perturbed Hamiltonian $H$ satisfies ETH and all states thermalize.

\medskip

\end{abstract}

\section{Introduction}
\label{sec:intro}

We consider an isolated macroscopic quantum system $S$ in a pure state $\psi\in\Hilbert$ evolving unitarily, $\psi_t=e^{-iHt}\psi_0$, for simplicity with $\psi_0$ (and thus $\psi_t$) in a micro-canonical ``energy shell'' $\Hilbert_\mc$, the spectral subspace of $H$ corresponding to energies in a small interval $[E-\Delta E,E]$. There are different concepts of thermal equilibrium of quantum systems \cite{RDO08,Reimann08,LPSW09,GLMTZ10, tas16, GogEis16,Gem15,Ueda,lanford1972approach,haag1973asymptotic,sukhov1983convergence,jakvsic2024approach,sugimoto2023bounds} (for a comparison of some, see \cite{GHLT15,GHLT16}), and one important concept for an isolated system $S$ is that its state $\psi$ lies, at least approximately, in a certain subspace $\Hilbert_\eq\subseteq \Hilbert_\mc$ containing the pure states that ``look macroscopically like thermal equilibrium states.''  Following \cite{GHLT15,GHLT16}, 
we call this concept ``macroscopic thermal equilibrium'' (MATE) \cite{GLMTZ10, tas16, Ueda, GLTZ19, SWGW23} and speak of ``macroscopic thermalization'' if $\psi_t$ reaches MATE sooner or later (even though $\psi_0$ may be far from MATE) and stays there for most of the time.
For brevity, we will drop the adjective ``macroscopic'' and just speak of ``thermal equilibrium'' and ``thermalization'' in the following. 

To take a concrete example, consider a gas of many but finitely many particles in a box. In the classical case it is well known, going back to Boltzmann, that the majority of microstates on some energy shell look macroscopically like a gas in thermal equilibrium, i.e., have uniform empirical position distribution and Maxwellian empirical momentum distribution. Moreover, it is expected (but very hard to prove) that the system starting from the majority of those fewer microstates corresponding to some non-equilibrium macrostate (say all particles start in the same half of the box) will also thermalize after some time, i.e., look macroscopically like a gas in thermal equilibrium for most later times. Note that this concept of thermalization does not require a heat bath, nor infinite system size, nor randomness in the evolution. The framework sketched above captures the analogous question for quantum systems, and our explicit example is indeed the (perturbed) free Fermi gas in a box.

A natural question is under which conditions on $\Hilbert_\eq$, $H$, and $\psi_0$ the system will thermalize. An observation made by Goldstein et al.~\cite{GLMTZ10} and Tasaki \cite{tas16} is that if~$H$ has non-degenerate spectrum and satisfies the appropriate version of the eigenstate thermalization hypothesis (ETH) \cite{Deutsch91,Srednicki94}, i.e., if 
\be\label{ETH1}
\text{every eigenvector of $H$ is in MATE,}
\ee
then \emph{every} $\psi_0$ thermalizes. Goldstein et al.~\cite{GLMTZ10} further proved that if $\dim \Hilbert_\eq/\dim \Hilbert_\mc$ is close to 1 (which is usually satisfied in practice), and if we take $H$ to be a random matrix with unitarily invariant distribution (or, equivalently, with an eigenbasis that is uniformly (Haar) distributed over all orthonormal bases and independent of the eigenvalues) and non-degenerate eigenvalues, then it satisfies the ETH \eqref{ETH1} with probability close to 1. (A similar result was obtained by Reimann \cite{Reimann2015}.) An observation that we add in Proposition~\ref{prop:ETH} below is that \eqref{ETH1} alone guarantees that every $\psi_0$ thermalizes, even for Hamiltonians with highly degenerate spectra.  

The present paper is inspired particularly by recent works of Tasaki \cite{T24,T24b} in which he focused on specifying concrete $\Hilbert_\eq$ and $H$ and proving for them that every $\psi_0$ thermalizes. The goal of proving thermalization for specific Hamiltonians brings into focus difficulties arising from highly degenerate eigenvalues; this paper is mainly about ways to deal with these difficulties. 

Concretely, Tasaki \cite{T24,T24b} (and earlier Shiraishi and Tasaki \cite{ST23}) considered $N\gg 1$ free non-relativistic fermions (``the free Fermi gas'') on  the 1d lattice $\Lambda\coloneqq \ZZZ/L\ZZZ$ with $L> N$ sites, which defines a Hilbert space $\Hilbert$ and a Hamiltonian $\HfF$, and took $\Hilbert_\eq$, as a simple model, to comprise the states for which the number of particles in a subinterval $\Gamma\subset\Lambda$ of the lattice lies within a suitable tolerance of $N|\Gamma|/L $.\footnote{A more realistic model of $\Hilbert_\eq$ would involve (a)~not only one region $\Gamma$ but every (suitably coarse-grained) macroscopic region in space, (b)~the (coarse-grained) distribution of momenta, and (c)~other macroscopic observables such as total spin. 
Item (a) can be implemented rather easily. Indeed, in \cite{T24,T24b}, Tasaki  partitioned the lattice into a (not too large) number of subintervals $\Gamma_i$ (thought of as macroscopic regions) and required, as the definition of $\Hilbert_\eq$, that the number of particles in each $\Gamma_i$ lies within suitable tolerances of $N|\Gamma_i|/L$, so the coarse grained empirical distribution of particles is approximately uniform in 1d physical space. Since this setup can be dealt with mathematically in much the same way as just considering a single $\Gamma$ (see Remark~\ref{rem:density}), we will stick here with the simpler model. Item (b), on the other hand, is pointless for the free Fermi gas, since individual momenta are conserved, and the treatment of the interacting Fermi gas is far beyond the scope of this paper.} 
Since for this simple model, the choice of $\Hilbert_\eq$ involves no conditions on the energy, the restriction to $[E-\Delta E,E]$ plays no role, so we can simply take $\Hilbert_\mc=\Hilbert$.
It is easy to see that $\dim\Hilbert_\eq/\dim\Hilbert$ is indeed close to 1. However, $\HfF$ has highly degenerate eigenvalues, and for this reason Shiraishi and Tasaki considered a perturbation $H_\theta^\mathrm{fF}$ of $\HfF$ by a small magnetic flux $\theta$ through the ring, which removes all degeneracies if $L\geq 3$ is  prime. Correspondingly, they proved thermalization of any initial state $\psi_0$ under the evolution generated by  $H_\theta^\mathrm{fF}$. The question whether a similar statement holds also for $\HfF$ was left open.

One of our main results (Theorem~\ref{thm2} in Section~\ref{sec:resultsfermions}) proves the ETH \eqref{ETH1} relative to this $\Hilbert_\eq$ also for the unperturbed free fermion Hamiltonian $\HfF$ in one spatial dimension. As a corollary we conclude, using the observation explained above, that also under the evolution of $\HfF$ every initial state $\psi_0$ thermalizes.

We also present results in another direction: Consider for an orthonormal basis (ONB) $B$ the condition that
\be\label{ETHONB1}
\text{every $\phi\in B$ lies in MATE.}
\ee
It turns out that there is \emph{one} eigen-ONB $B_1$ of $\HfF$ that is particularly good in several ways (see Proposition~\ref{prop: ETH higherdim} in Section~\ref{sec:results}):
\begin{itemize}
    \item[(i)] $B_1$ satisfies \eqref{ETHONB1} with much smaller error bounds (i.e., smaller deviations from $\Hilbert_\eq$) than some other eigen-ONBs; in fact like $e^{-N}$ instead of a negative power of $N$;
    \item[(ii)] it is easier to prove \eqref{ETHONB1} for $B_1$ than for other eigen-ONBs;
    \item[(iii)] in higher dimensions (i.e., considering a lattice $\ZZZ^d/L\ZZZ^d$ with $d>1$), we could find a proof of \eqref{ETHONB1} for (the analog of) $B_1$, but not for \emph{all} eigen-ONBs of (the analog of) $\HfF$.
\end{itemize}
This situation motivates us to propose a strategy for dealing with degenerate Hamiltonians for which some, but not every, eigen-ONB $B$ satisfies \eqref{ETHONB1}. Let us call such a general Hamiltonian $H_0$. (Note that violations of \eqref{ETH1} imply the existence of some $\psi_0$ that will not thermalize: for example, eigenstates that are not initially in MATE will never reach MATE because they are stationary.) 
If $H_0$ is only moderately degenerate, then an elementary estimate (see Corollary~\ref{cor1} in Section~\ref{sec:background}) shows that\eqref{ETHONB1} for one eigenbasis would enforce \eqref{ETHONB1} for every other eigenbasis with moderately worse error bounds.
However, in the example of the free Fermi gas in $d\geq 1$ space dimensions, the degeneracy is around $2^{Nd}$ and thus too high. 
For this reason, we propose to consider a perturbation of $H_0$,\footnote{Still, the considerations stay fully rigorous. In particular, they do not involve neglecting higher-order terms in a series expansion (as the word ``perturbation'' might suggest in some contexts).} but (unlike \cite{ST23,T24,T24b}) a \emph{random} perturbation
\be\label{HH0V}
H=H_0+\lambda V\,,
\ee
thought of as a \emph{generic} perturbation. We argue in Section~\ref{sec:generic} that it is physically appropriate to consider a generic perturbation. In fact, as soon as $V$ has continuous probability distribution, $H$ has non-degenerate eigenvalues and eigenvalue gaps with probability 1 for every $0<\lambda<\lambda_0$ for suitable $\lambda_0$.\footnote{Since we could not find a good reference for this fact, we have formulated this fact as Lemma~\ref{lem:1} in Section~\ref{sec:background} and included a proof in Section~\ref{sec:proofs}.} So, generic (arbitrarily small) perturbations remove the degeneracy of $H$. 

Moreover, and this is the content of our other main result (Theorem~\ref{thm1 new} in Section~\ref{sec:results}), under the assumption that the distribution of $V$ is invariant under all unitaries or at least those commuting with $H_0$, most $H$ with $H_0$ having one eigenbasis whose elements are in MATE satisfy the ETH \eqref{ETH1}, and thus for most perturbations all states thermalize.

Apart from the general, abstract argument expressed in Theorem~\ref{thm1 new}, we also include the application to the free Fermi gas (Proposition~\ref{prop: ETH higherdim} and Corollary~\ref{cor: ETH higherdim} in Section~\ref{sec:resultsfermions}): the existence of an eigen-ONB $B_1$ of $\HfF$ in MATE (with a tolerance $\varepsilon$ that is exponentially small in $N$) can be established in any dimension $d$ of physical space; 
the thermalization of all states for $H=\HfF+\lambda V$  then follows from Theorem~\ref{thm1 new}.

Note that analogous results for the classical free gas of $N$ particles have been obtained in \cite{Beck10,dBP17,Beck17}. There, no perturbation is needed to prove thermalization of the gas for long times in any spatial dimension, but only for most and not all initial states.\footnote{It is obvious that for the classical free gas not all initial states can thermalize, not even in the sense of becoming spatially homogeneous.}
For the one-dimensional quantum mechanical free Fermi gas we obtain an even stronger result, namely thermalization for all initial states. For higher dimensions, however, we need small generic perturbations of $\HfF$ to infer thermalization for all initial states.

It should be noted that neither \cite{ST23,T24,T24b} nor our results provide meaningful estimates of the time required to reach thermal equilibrium in the (perturbed) free Fermi gas. For the classical gas such estimates have been obtained in \cite{Beck17,dBP17}. A more quantitative understanding of the non-equilibrium dynamics of one-dimensional integrable quantum gases is provided by so called Quantum Generalized Hydrodynamics, see, e.g., \cite{essler2016quench,ruggiero2020quantum}. However, this theory is not rigorous yet and applies only to rather special initial  states.

Finally, let us mention another non-trivial application of  Theorem~\ref{thm1 new}. Shortly after the first version of our paper appeared as a preprint, Hal Tasaki \cite{T24c} realized that   Theorem~\ref{thm1 new} can be applied to the Ising model in two dimensions below the critical temperature. Roughly speaking, he proves that any initial state in a given highly degenerate eigenspace of the Hamiltonian thermalizes under most slightly perturbed dynamics in the sense that the macroscopic magnetization approaches and remains very close to the corresponding microcanonical expectation value. In this system it is indeed the case that some eigenbases of the unperturbed Hamiltonian satisfy the ETH and others do not. 

This paper is structured as follows. In Section~\ref{sec:background}, we provide an overview of the background and provide further  motivation. In Section~\ref{sec:results}, we state our main results. In Section~\ref{sec:proofs}, we give the proofs. In Section~\ref{sec:conclusions}, we conclude.

\section{Motivation}
\label{sec:background}

In this section, we give some more details about the considerations outlined in the introduction.

\subsection{MATE and ETH}

We will only operate within one energy shell $\Hilbert_\mc$ and pretend that this subspace remains unchanged even when we vary the Hamiltonian; 
we take $\Hilbert_\mc$ to be ``the'' Hilbert space of the system $S$ and simply write $\Hilbert$ for it. We take for granted that $\Hilbert$ has finite dimension. Following von Neumann \cite{vonNeumann29}, we regard macroscopic observables as given which are suitably coarse-grained so that they commute with each other and their eigenvalues are rounded to the macroscopic resolution. Then $\Hilbert_\eq$ can be thought of as one of their simultaneous eigenspaces, with eigenvalues given by the thermal equilibrium values \cite{GLMTZ10,GLTZ19}. In our mathematical result, $\Hilbert_\eq$ could be any subspace, although it will play a role that $\Hilbert_\eq$ has most of the dimensions of $\Hilbert$.

We use the notation $P_\eq$ for the projection to $\Hilbert_\eq$ and
\be
\SSS(\Hilbert)=\bigl\{\psi\in\Hilbert:\|\psi\|=1\bigr\}
\ee
for the unit sphere in any given Hilbert space $\Hilbert$. 
Let $u$ be the uniform (normalized surface area) measure on $\SSS(\Hilbert)$; a $u$-distributed vector will also be said to be ``purely random.'' 
When saying that ``the statement $S(\psi)$ is true for $(1-\varepsilon)$-most $\psi\in\SSS(\Hilbert)$,'' we mean that 
\be
u\{\psi\in\SSS(\Hilbert):S(\psi) \text{ holds}\}\geq 1-\varepsilon\,.
\ee
Analogously, when saying that ``the statement $S(t)$ is true for $(1-\delta)$-most $t\in[0,\infty)$,'' we mean that
\be
\liminf_{T\to\infty} \frac{1}{T}\bigl| \{t\in[0,T]: S(t)\text{ holds}\} \bigr| \geq 1-\delta\,,
\ee
where $|\{\cdot\}|$ means the length (Lebesgue measure) of the set $\{\cdot\}$.

Since $\SSS(\Hilbert_\eq)$ is a null set in $\SSS(\Hilbert)$ relative to $u$, we regard a $\psi\in\SSS(\Hilbert)$ as being in thermal equilibrium whenever it lies in the set
\be
\mathrm{MATE}_\varepsilon =\bigl\{\psi\in\SSS(\Hilbert): \|P_\eq \psi\|^2 \geq 1-\varepsilon \bigr\}\,,
\ee
once we have chosen the desired tolerance $\varepsilon>0$. Correspondingly, $\mathrm{ETH}_\varepsilon$ denotes the condition on a Hamiltonian $H$ that 
\be\label{ETH}
   \forall \text{normalized eigenvector $\phi$ of $H$:}\quad \phi \in \mathrm{MATE}_{\varepsilon}.
\ee
In terms of the spectral decomposition $H = \sum_{e} e\, \Pi_{e}$, this condition is equivalent to $\| \Pi_{e} P_\noneq \Pi_e \| \leq {\varepsilon}$ for all eigenvalues $e$, where $P_\noneq \coloneqq I-P_\eq$ and $\|\!\cdot\!\|$ denotes the operator norm.
The statement that the ETH \eqref{ETH1} implies that every $\psi_0$ thermalizes can be formulated rigorously as follows.

\begin{prop}\label{prop:ETH}
    Suppose $\dim\Hilbert=:D<\infty$, $\Hilbert_\eq$ is a subspace and $P_\eq$ the projection to it, $\varepsilon,\delta>0$, the operator $H$ on $\Hilbert$ is self-adjoint, and satisfies $\mathrm{ETH}_{\varepsilon\delta}$. 
    Then for every $\psi_0\in\SSS(\Hilbert)$ and $(1-\delta)$-most $t\in[0,\infty)$,
    \be\label{MATEt}
    \psi_t \in \mathrm{MATE}_\varepsilon \,.
    \ee
\end{prop}

All proofs are given in Section~\ref{sec:proofs}. 
Note that the inaccuracy~($\varepsilon\delta$) assumed in the ETH must be smaller than the inaccuracy~($\varepsilon$) desired for the thermalization of $\psi_t$ by a factor given by the tolerance~($\delta$) desired for the notion of ``most $t$.'' Note that for non-degenerate $H$ the statement of 
Proposition~\ref{prop:ETH} was contained in \cite{GLMTZ10}  as a step in a proof and in \cite{tas16} as Theorem~7.1. However, to our knowledge the observation that the ETH can also be used for degenerate Hamiltonians in the form \eqref{ETH} to derive thermalization of every initial state seems not to have been mentioned before in the literature. In particular, it entails that \emph{all} $\psi_0$ thermalize for the Hamiltonian $\HfF$ of the one-dimensional free Fermi gas, which satisfies \eqref{ETH} by Theorem~\ref{thm2} in Section~\ref{sec:resultsfermions}.

\subsection{The Problem About High Degeneracy}
\label{sec:problemdegen}

Suppose that a Hamiltonian $H_0$ is highly degenerate and possesses an eigen-ONB $B=(\phi_k)_k$ that satisfies \eqref{ETHONB1} or more precisely
\be\label{ETHONB}
\forall k: \quad \phi_k\in \mathrm{MATE}_\varepsilon.
\ee
Note that if $P_\noneq\neq0$, necessarily
\be\label{eq:eps_lowerbound}
\varepsilon \geq 1/\dim \Hilbert,
\ee
which follows from considering the trace of $P_\noneq$ in the basis $B$.
The question is how much the tolerance $\varepsilon$ has to be increased to ensure that all eigenvectors are in MATE (or, equivalently, $H_0$ satisfies ETH).
The following elementary result for matrices helps us answer this question:
\begin{lemma}\label{lea1}
    Let $\varepsilon>0$. If a positive semi-definite $D\times D$ matrix $M$ has all diagonal entries $\leq \varepsilon$, then
    \be
    \|M\|\leq \varepsilon D.
    \ee
    The bound is sharp, i.e., there exists $M$ for which equality holds.
\end{lemma}

By applying this lemma to each $\Pi_{e} P_\noneq \Pi_e$, we can guarantee ETH~\eqref{ETH1} with a larger error instead of $\varepsilon$:
\begin{cor}\label{cor1}
    Let $H_0$ be a Hamiltonian with maximal degeneracy $D_E$ in the $D$-dimensional Hilbert space $\Hilbert$.
    If $P_\eq$ is any projection such that $H_0$ has an eigen-ONB $(\phi_k)_k$ satisfying \eqref{ETHONB}, then
    \be
        \forall \text{normalized eigenvector $\phi$ of $H_{0}$:}\quad \phi \in \mathrm{MATE}_{\varepsilon D_E}.
    \ee
    The bound is sharp in the sense that for any $H_{0}$ and any $\varepsilon\: (\geq 1/D)$, there exists $P_\eq$ and an eigenvalue $E$ such that $\| \Pi_E P_\noneq \Pi_E \| = \min\{\varepsilon D_E, 1\}$ and $(\phi_k)_k$ satisfies \eqref{ETHONB}.
\end{cor}

This means that the error bound we can guarantee in the ETH~\eqref{ETH} is $\varepsilon D_E$ instead of $\varepsilon$, but for our example of $H_0$ (the free Fermi gas) in $d\geq 1$ dimensions,  
\be
D_E \geq 2^{Nd}
\ee
(see Proposition~\ref{prop: ETH higherdim}), which is so large that the bound $\varepsilon D_E$ is no longer small and thus becomes useless.

Now let $P_\eq$ be any projection, $P_\noneq=I-P_\eq$, and let $H_0$ be a degenerate Hamiltonian for which one eigen-ONB satisfies \eqref{ETHONB} but others do not necessarily. By Corollary~\ref{cor1}, some eigenvectors can deviate from $\Hilbert_\eq$ by $\varepsilon D_E$, which need not be small as $D_E$ can be exponentially large in the particle number $N$. On the other hand, even if $D_E$ is that large, a \emph{typical} unit vector is not that bad:

\begin{prop}\label{prop:mostphi}
    Let $\varepsilon, \eta >0$, $\Hilbert_e$ be an eigenspace of $H_0$ with dimension $D_e$, and suppose that one eigen-ONB of $H_0$ satisfies \eqref{ETHONB}. Then for $\delta=2\exp(-C \eta^2 D_e)$ with $C=2/9\pi^3$, $(1-\delta)$-most $\phi\in\SSS(\Hilbert_e)$ lie in $\mathrm{MATE}_{\varepsilon+\eta}$.
\end{prop}

The following proposition shows that most eigenbases of an arbitrary eigenspace $\Hilbert_e$ of $H_0$ are in MATE.
\begin{prop}\label{prop: most eigenbases}
    Let $\varepsilon,\eta>0$, $\Hilbert_e$ be an eigenspace of $H_0$ with dimension $D_e$, let $(\phi_k)_{k=1}^{D_e}$ be a random ONB of $\Hilbert_e$ and suppose that one eigen-ONB of $H_0$ satisfies \eqref{ETHONB}. Then,
    \begin{align}
        \mathbb{P}\Bigl(\exists k\leq D_e: \phi_k\notin \textup{MATE}_{\varepsilon+\eta}\Bigr) \leq 2D_e \exp\left(-C\eta^2 D_e\right),
    \end{align}
    where $C=2/9\pi^3$.
\end{prop}

We can now combine Corollary \ref{cor1} and Proposition \ref{prop: most eigenbases} to show that most eigenbases of $H_0$ lie in MATE.
\begin{prop}\label{prop: most eigenbases H0}
    Let $\varepsilon>0$, $C=2/9\pi^3$, let $(\phi_k)_{k=1}^{D}$ be a random eigen-ONB of $H_0$ and suppose that one eigen-ONB of $H_0$ satisfies \eqref{ETHONB}. Then, for all $\eta>0$ such that $C \eta^3 \geq \varepsilon$,
    \begin{align}
        \mathbb{P}\Bigl(\exists k\leq D: \phi_k\notin \textup{MATE}_{\varepsilon+\eta}\Bigr) \leq 2D \exp\left(-C \eta^3/\varepsilon\right).
    \end{align}
\end{prop}
For the free Fermi gas with $N$ particles, we shall see in Section~\ref{sec:resultsfermions} that both $D$ and $1/\varepsilon$ are exponentially large in $N$.
Therefore, for $\eta$ of order $\varepsilon^\alpha$ for any $0<\alpha<1/3$, the right-hand side of the inequality above vanishes as $N\to \infty$.
In particular, for large $N$ most eigenbases satisfy $\textrm{MATE}_\eta$ with $\eta$ exponentially small in $N$.

\subsection{Consequences of Generic Perturbations}

The following lemma provides a precise formulation of the intuitively rather obvious statement that a generic perturbation will lift the degeneracy of eigenvalues and eigenvalue gaps. (In this paper, it will play no role that the eigenvalue gaps are non-degenerate, but we take note of this fact for the sake of completeness,  as the degeneracy of the eigenvalue gaps is expected to determine the time scales on which thermalization takes place.) 

\begin{lemma}\label{lem:1}
    Let $V$ be a random matrix whose distribution is continuous in the space of Hermitian $D\times D$ matrices. Then with probability 1, there exists $\lambda_0>0$ such that for all $\lambda\in (0,\lambda_0)$, the Hamiltonian $H=H_0+\lambda V$ has non-degenerate eigenvalues and eigenvalue gaps.
\end{lemma}

Now we consider $H=H_0+\lambda V$, where the distribution of the random matrix $V$ is continuous and invariant under all unitaries commuting with $H_0$. Then the eigen-ONB of $H$ (which is unique up to phase factors because $H$ is non-degenerate) will be arbitrarily close, for sufficiently small $\lambda$, to some eigen-ONB of $H_0$. In fact, the eigen-ONB of $H$ (with suitably chosen phase factors) converges as $\lambda\to0$ (for fixed $V$) to an eigen-ONB $(\chi_k)_k$ of $H_0$.\footnote{This fact also entails that if we subdivide the energy axis into ``micro-canonical'' intervals, then for sufficiently small $\lambda$, each $\Hilbert_\mc$ obtained from $H_0$ stays invariant during any time interval $[0,T]$ to an arbitrary degree of precision under the time evolution generated by $H$, with the consequence that each micro-canonical subspace can be treated separately, and our simplifying assumption that $\Hilbert_\mc$ is invariant caused no harm.} Which one? That depends on $V$. Due to the unitary invariance, the ONB $(\chi_k)_k$ is, in each eigenspace $\Hilbert_e$ of $H_0$, uniformly distributed among the ONBs of $\Hilbert_e$. 
If most eigenbases of $H_0$ lie in MATE, as for the free Fermi gas,  there is ETH and thus thermalization for most perturbations $V$, see Theorem~\ref{thm1 new} below.

\subsection{Physical Relevance of Generic Perturbations}
\label{sec:generic}

In physics, we sometimes make models (e.g., write down a formula for the Hamiltonian) and sometimes consider generic situations (e.g., consider a random Hamiltonian). When we have proved that most Hamiltonians relative to a particular distribution have a certain property $P$, then this still leaves open whether the true Hamiltonian has this property. On the other hand, models involve idealizations and simplifications, and therefore are not necessarily realistic. So, when we have proved that a model has property $P$, this also leaves open whether the true Hamiltonian has this property. 

In order to increase the reliability of mathematical results about $P$, one can try to add more realism. For a model, this might mean to add corrections, such as relativistic corrections, previously neglected interaction terms between the particles of the system, or interactions with the outside that the system is not perfectly shielded from (such as gravitational interactions). Of course, this will often make the model intractable. For a random Hamiltonian, on the other hand, increasing realism may mean making the distribution narrower, either by conditioning on properties $P'$ that we believe the true Hamiltonian has (e.g., symmetries) or by choosing a distribution near some $H_0$ that we believe the true Hamiltonian is close to. The latter strategy is, in fact, a kind of combination of the two strategies of considering a model $H_0$ and considering a random Hamiltonian. 

Our assumption that the distribution of the perturbation $V$ is invariant under unitaries (at least those commuting with $H_0$) is motivated (i)~by the thought that, since we are considering very small perturbations, many different kinds of interaction with the environment may contribute to $V$ and (ii)~by the facts that this would be the case for the simplest distributions of $V$, such as the Gaussian unitary ensemble GUE, and that this allows us to answer the question how many $\psi_0$ will thermalize. The fact that a random $V$ with unitarily invariant distribution involves super-long-range super-multi-body interactions makes it seem unrealistic. On the other hand, when we consider \emph{very} weak perturbations, then already the fact that no system is \emph{exactly} closed becomes relevant---that every system is slightly interacting with an environment such as a gas of photons (or of gravitational waves etc.). 
The reason we are considering closed systems (that evolve in a Hamiltonian rather than Lindbladian way) is that being open is unnecessary for thermalization. And yet, when it comes to arbitrarily weak perturbations, a weak interaction with an environment may be expected to have a similar effect as a weak generic perturbation of $H$. After all, if the particles of the environment (say, photons)
are entangled with each other, then distant parts of the system will effectively interact with each other by interacting with different entangled photons. Thus, the system's evolution  seems quite similar to a unitary model in which every part of the system is weakly interacting with every other.

A relevant trait of our results in this paper is that they apply to random perturbations of $H_0$ that are \emph{arbitrarily weak}. Such results can be regarded as stating an \emph{instability} of a property $P$: For example, being degenerate is an unstable property in the sense that every degenerate $H_0$ possesses a neighborhood in the space of self-adjoint operators in which the degenerate operators form a null set. It is therefore not believable that the true Hamiltonian is degenerate, given that it is close to $H_0$. Also deterministic corrections to $H_0$ may be expected to break the degeneracy, but again it may be intractable to prove this.

The upshot is that assuming an arbitrarily small generic perturbation may be quite realistic after all. The typical behavior of such a perturbation may be a pretty good prediction of the behavior of the true Hamiltonian.

\section{Main Results}
\label{sec:results}

In this section, we present and discuss our main results. In Section~\ref{sec:genham}, we state our results for general Hamiltonians and in Section~\ref{sec:resultsfermions} we apply them to the free Fermi gas.

\subsection{For General Hamiltonians\label{sec:genham}}

\begin{thm}\label{thm1 new}
Let  $\Hilbert$ be a Hilbert space with $D\coloneqq \dim\Hilbert<\infty$, $\Hilbert_\eq$ any subspace, $P_\eq$ the associated orthogonal projection, and $P_\noneq \coloneqq  I-P_\eq$. 
Let $H_0\in\mathcal{L}( \Hilbert)$ be self-adjoint and assume that $H_0$ has an orthonormal eigenbasis $( \phi_k)_{k\in\{1,\ldots,D\}}$ such that for all $k$ the eigenvector $ \phi_k\in\mathrm{MATE}_{\varepsilon}$ with respect to $\Hilbert_\eq$ for some $\varepsilon>0$ (i.e., $\| P_{\noneq} \phi_k\|^2<\varepsilon$ for all $k$).
For $\lambda\in\mathbb{R}$ let $H\coloneqq H_0+\lambda V$, where $V$ is a self-adjoint operator drawn randomly from a continuous distribution invariant under conjugation with all unitaries commuting with $H_0$.
Furthermore, for $\psi_0 \in \mathbb{S}(\Hilbert)$ let $\psi_t\coloneqq  e^{-iHt}\psi_0$, let $C=2/9\pi^3$ and for $\eta>0$ satisfying $C \eta^3 \geq \varepsilon$ let $\beta=2D\exp(-C \eta^3 /\eps)$.\\
Then for all $0<\alpha<1$ there is a $\lambda_0>0$ such that for all $0<\lambda<\lambda_0$ and $(1-\alpha)(1-\beta)$-most  $V$, $H_0+\lambda V$ satisfies $\mathrm{ETH}_{\eps+2\eta}$ and hence for all $\delta>0$ and all $\psi_0\in\mathbb{S}(\Hilbert)$ for $(1-\delta)$-most $t\in[0,\infty)$,
    \begin{align}
        \psi_t \in \mathrm{MATE}_{\frac{\eps+2\eta}{\delta}}.
    \end{align}
\end{thm}

\begin{rmk}[Conditions on the distribution of $V$]
While the condition of invariance under the unitaries commuting with $H_0$ is the minimal condition we need to mathematically prove the consequence, the most relevant cases in practice are perhaps those in which the distribution is invariant under \emph{all} unitaries, as is the case for the Gaussian unitary ensemble (GUE) in which the entries of $V$ are (up to Hermitian symmetry) i.i.d.\ complex Gaussian random variables.
\end{rmk}

\begin{rmk}[Previous version of Theorem~\ref{thm1 new}]\label{rem:improved}
Generally, for a system consisting of $N$ particles, its dimension $D$ is of order $e^N$. As we discussed below Proposition~\ref{prop: most eigenbases H0}, for the free Fermi gas with $N$ particles, $\beta$ in Theorem~\ref{thm1 new} is small if $\eta$ is of order $\varepsilon^\alpha$ with $0<\alpha<1/3$ as $\varepsilon$ is exponentially small in $N$. However, if $N\varepsilon$ is not small, then $\beta$ is not small and Theorem~\ref{thm1 new} yields nothing interesting as the set of admissible perturbations $V$ becomes small or even of measure zero. In this case we can still prove thermalization (with a worse bound) for most perturbations $V$ for most initial wave functions $\psi_0\in\mathbb{S}(\Hilbert_\nu)$ from an arbitrary subspace $\Hilbert_\nu\subset\Hilbert$. For the precise version and proof of this statement we refer the reader to an earlier  version of this paper available at \url{http://arxiv.org/abs/2408.15832v3}.
\end{rmk}

\subsection{For the Free Fermi Gas}
\label{sec:resultsfermions}

In this subsection we discuss the concrete example of the free, non-relativistic Fermi gas of $N$ particles on a $d$-dimensional lattice $\Lambda\coloneqq \{1,\ldots,L\}^d$ with periodic boundary conditions, where $L\in\mathbb{N}$ and $d\geq 1$. We will in particular show for $d=1$ that all eigen-ONBs satisfy the ETH, so  Proposition~\ref{prop:ETH} applies, and for $d\geq1$ that Theorem~\ref{thm1 new} applies.

The Hamiltonian is given by
\begin{align}
    \HfF \coloneqq  - \sum_{\substack{x,y \in \Lambda\\ \mathrm{dist}(x,y)=1}} c_x^\dagger c_y, \label{eq:H0}
\end{align}
where $c_x$ and $c_x^\dagger$ denote the annihilation and creation operators of a fermion at site $x\in\Lambda$. ($H_0+2NdI$ is the negative discrete Laplacian.) 
The relevant Hilbert space 
is the $N$-particle sector of fermionic Fock space, i.e.,  $ \Hilbert \simeq \mathbb{C}^D$ with $D=\binom{L^d}{N}$. 

Like Shiraishi and Tasaki \cite{ST23,T24,T24b}, we will use a highly simplified model of ``thermal equilibrium'' defined only in terms of the spatial distribution of particles. In particular, the restriction to a micro-canonical energy shell is not relevant for us in the following. We will discuss extensions to more realistic models of thermal equilibrium in Remark~\ref{rem:density}. Choose any spatial region  $\Gamma \subseteq \Lambda$, let 
\be \mu\coloneqq  |\Gamma|/|\Lambda|
\ee
be its relative size, and let
\begin{align}
    N_\Gamma \coloneqq \sum_{x \in \Gamma} c_x^\dagger c_x \label{eq: NGamma}
\end{align}
be the number operator of the particles in $\Gamma$. 
Throughout this section we define the equilibrium subspace  
$\Hilbert_{\eq,\eta}$ 
for a given threshold $\eta>0$ as the spectral subspace of $N_\Gamma$ specified by the condition $\left|\frac{N_{\Gamma}}{N} -\mu\right|\leq \eta$, i.e.,
\begin{align}\label{Hgamma}
    P_{\eq,\eta} \coloneqq  \mathbf{1}_{[N(\mu-\eta),N(\mu+\eta)]}(N_\Gamma)\,,
\end{align}
and set $P_{\noneq,\eta}\coloneqq I-P_{\eq,\eta}$.
Thus, $\Hilbert_\mathrm{eq,\eta}$ contains those states $\psi$ for which the Born distribution of $N_\Gamma/N$ is supported in an $\eta$-neighborhood around $\mu$. Note that this is a much stronger condition than just requiring that the expectation value of $N_\Gamma/N$ in a state $\psi$ lies in this interval.
In Remark~\ref{rem:density} below we discuss a more realistic definition of $P_{\eq,\eta}$ that takes into account not only the number of particles  in one region $\Gamma$, but a coarse-grained density all over $\Lambda$.

For $d=1$ we can show that $\HfF$ satisfies the version \eqref{ETH1} of the ETH, i.e., that every eigenvector of $\HfF$ is close to  $\Hilbert_{\eq,\eta}$ for large $N$, assuming that $\Gamma$ is an interval.

\begin{thm}[ETH for the free Fermi gas in 1d]\label{thm2}
    Let $d=1$,  $L$ prime, $46\leq N< L/4$,  $\Gamma\subset\Lambda$ an interval, $\eta > \frac{2(\ln N + 1)}{N}$, and $\Hilbert_{\eq,\eta}$   and   $P_{\noneq,\eta} $ as above.
    Then every normalized eigenstate $\phi$ of $\HfF$ satisfies
    \begin{align}\label{eqthm2}
	\bigl\| P_{\noneq,\eta}\phi\bigr\|^2 \leq \frac{32\ln N}{\eta^2 N}.
	\end{align}
\end{thm}
The condition in Theorem~\ref{thm2} that $L$ must be prime guarantees that the eigenvalues of the one-body Hamiltonian on such a chain are rationally independent (see \cite{ST23} and similar arguments in \cite{keating2015spectra}). As a consequence, the only degeneracies in the many-body spectrum arise from the fact that one-body eigenstates with momentum $k$ and $-k$ have the same energy.  The latter degeneracy was shown by Shiraishi and Tasaki to be removed by piercing the ring with a small magnetic flux. Theorem~\ref{thm2} shows that even if these degeneracies are not removed, $\HfF$ still satisfies the ETH.

Together with Proposition~\ref{prop:ETH}, Theorem~\ref{thm2} implies that all initial states reach MATE:

\begin{cor}[Thermalization of the  free Fermi gas in 1d]\label{cor:free1d}
 Let $d=1$, $N\geq 46$, and  let $L$, $\Gamma$, $\eta$, $\Hilbert_{\eq,\eta}$ and $ P_{\eq,\eta}$ be as in Theorem~\ref{thm2}.
 Let $\varepsilon,\delta>0$ be such that $\varepsilon\delta\geq \frac{32\ln N}{\eta^2 N}$.
     Then for every $\psi_0\in\SSS(\Hilbert)$ and $(1-\delta)$-most $t\in[0,\infty)$, $\psi_t \coloneqq  e^{-i\HfF t}\psi_0$ satisfies
    \be 
      \left\|P_{\noneq,\eta}\psi_t\right\|^2 <\varepsilon\,,~~\text{i.e.,}~~\psi_t\in\mathrm{MATE}_\varepsilon\,.
    \ee
\end{cor}

For general $d\geq 1$, we can still prove MATE for one eigenbasis $B_1$ of~$\HfF$, in fact with better bounds. To define this eigenbasis, we need to introduce some notations first.

For odd $L$ let
\begin{align}
		\mathcal{K} \coloneqq  \left\{\frac{2\pi}{L}\nu\Bigg| \nu \in \left\{0,\pm 1,\dots, \pm \frac{L-1}{2}\right\}^d\right\},
\end{align}
and for even $L$ let
\begin{align}
		\mathcal{K} \coloneqq  \left\{\frac{2\pi}{L}\nu\Bigg| \nu \in \left\{0,\pm 1,\dots, \pm \Bigl( \frac{L}{2}-1 \Bigr), \frac{L}{2}\right\}^d\right\}.
\end{align}
For $k\in\mathcal{K}$ we define
\begin{align}
    a_k^\dagger \coloneqq  \frac{1}{L^{d/2}} \sum_{x\in\Lambda} e^{ik\cdot x} c_x^\dagger.
\end{align}
Let $\mathcal{K}^N_{\neq}$ be the set of $k=(k_1,\ldots,k_N)\in\mathcal{K}^N$ such that $k_i\neq k_j$ for all $i\neq j$. The permutation group $S_N$ acts on $\mathcal{K}^N_{\neq}$ via $(k_1,\ldots,k_N)\mapsto (k_{\pi(1)},\ldots,k_{\pi(N)})$ for any $\pi\in S_N$. Let $\tilde{\mathcal{K}}^N \subset\mathcal{K}_{\neq}^N$ contain exactly one representative from each orbit, i.e., from each permutation class.
We define
\begin{align}
    B_1&\coloneqq \bigl\{|\Psi_k\rangle: k\in\tilde{\mathcal{K}}^N\bigr\}, \text{ where}\\
	|\Psi_k\rangle &\coloneqq  a_{k_1}^\dagger a_{k_2}^\dagger \dots a_{k_N}^\dagger |\Phi_{\textup{vac}}\rangle \label{eq: energy eigenstate highdim}
\end{align}
with $|\Phi_{\textup{vac}}\rangle$ the vacuum vector in Fock space.
The states $|\Psi_k\rangle$ are $\binom{L^d}{N}$ different eigenfunctions of the unperturbed Hamiltonian $\HfF$ and therefore form an orthonormal basis of $\Hilbert$.

For this eigenbasis $B_1$ of $\HfF$ we can prove MATE using similar methods as Tasaki \cite{T24} used in the case of the free fermion chain in one dimension.

\begin{prop}[MATE for one eigenbasis of the free Fermi gas, any $d$]\label{prop: ETH higherdim}
    Let $d\geq 1$, $\Gamma\subset\Lambda$ arbitrary,   $0<\eta < \frac{3}{2}\mu (1-\mu)$, and $\Hilbert_{\eq,\eta}$   and   $P_{\noneq,\eta} $ as in \eqref{Hgamma}.
    Then every eigenstate $\Psi_k\in B_1$ of~$\HfF$ given by \eqref{eq:H0} satisfies 
	\begin{align}\label{eq: ETH higherdim}
	\bigl\| P_{\noneq,\eta}\Psi_k\bigr\|^2 < 2 e^{-\frac{\eta^2}{3\mu(1-\mu)}N}.
	\end{align}
    Furthermore, if $N<L/2d$ then the maximal degree of degeneracy $D_E$ is at least $2^{Nd}$. 
\end{prop}

As an immediate consequence of Proposition~\ref{prop: ETH higherdim} and  Theorem~\ref{thm1 new} we obtain the following corollary: 

\begin{cor}[Thermalization of the perturbed free Fermi gas in any dimension]\label{cor: ETH higherdim}
 Let $d\geq 1$ and $\Gamma$,  $\eta$, $\Hilbert_{\eq,\eta}$, and $P_{\noneq,\eta}$ be as in Proposition~\ref{prop: ETH higherdim}. 
For $\lambda\in\mathbb{R}$ let $H\coloneqq \HfF+\lambda V$, where $V$ is drawn randomly from a continuous distribution invariant under conjugation with all unitaries commuting with~$\HfF$.
For $\psi_0\in \mathbb{S}(\Hilbert)$  let $\psi_t\coloneqq e^{-iHt}\psi_0$, let $C=2/9\pi^3$, $\varepsilon\coloneqq 2\exp(-\eta^2 N/(3\mu(1-\mu)))$ and for $\tilde{\eta}>0$ satisfying $C \tilde{\eta}^3 \geq\varepsilon$ let $\beta=2D\exp(-C \tilde{\eta}^3 /\varepsilon)$. Then for all $\delta>0$ and $0<\alpha<1$ there is a $\lambda_0>0$ such that for all $0<\lambda<\lambda_0$ and $(1-\alpha)(1-\beta)$-most $V$, every $\psi_0\in\mathbb{S}(\Hilbert)$ is such that for $(1-\delta)$-most $t\in[0,\infty)$,
\begin{align}
\left\|P_{\noneq,\eta}\psi_t\right\|^2 \leq \frac{\varepsilon + 2\tilde{\eta}}{\delta},~~\text{i.e.,}~~\psi_t\in\mathrm{MATE}_{(\varepsilon + 2\tilde{\eta})/\delta} \,.
\end{align}
\end{cor}

\begin{rmk}\label{rem:density}
In a similar way as in \cite{T24}, Theorem~\ref{thm2} and Proposition~\ref{prop: ETH higherdim} (and therefore also Corollary~\ref{cor:free1d} and Corollary~\ref{cor: ETH higherdim}) can easily be generalized to the situation in which we define equilibrium by requiring that in $m\in\mathbb{N}$ subsets of $\Gamma_i \subset \Lambda$ the fraction of particles is close to $\mu_i \coloneqq  |\Gamma_i|/|\Lambda|$. To this end, note that
\begin{align}
P_{\noneq,\eta} \coloneqq  P\left(\exists\, i=1,\dots,m: \;\Biggl|\frac{N_{\Gamma_i}}{N} - \mu_i \Biggr| > \eta\right) \leq \sum_{i=1}^m P\left(\Biggl|\frac{N_{\Gamma_i}}{N} - \mu_i \Biggr| > \eta\right),
\end{align}
where $P(\dots)$ denotes the orthogonal projection onto the specified subspace. Let $\mu_*\coloneqq \argmax_{\mu_1,\ldots,\mu_m} (\mu_i(1-\mu_i))$. Then it follows from Proposition~\ref{prop: ETH higherdim} that
\begin{align}
\left\| P_{\noneq,\eta} \Psi_k\right\|^2 \leq 2 \sum_{i=1}^m e^{-\frac{\eta^2}{3\mu_i(1-\mu_i)}N} \leq 2m\, e^{-\frac{\eta^2}{3\mu_*(1-\mu_*)}N}. \label{ineq: neq mult int} 
\end{align}
Thus, as long as $m$ is not too large and $N$ is large, the right-hand side in \eqref{ineq: neq mult int} is small. 
\end{rmk}

\section{Proofs}
\label{sec:proofs}

\subsection{Proof of Proposition~\ref{prop:ETH}}

Let $P_\noneq \coloneqq I-P_\eq$, let $\mathcal{E}$ be the spectrum of $H$, and $\Pi_e$ the projection to the eigenspace of $H$ with eigenvalue $e$. By denoting the long-time average by $\overline{X} \coloneqq \mathop{\lim}\limits_{T\to\infty} \frac{1}{T} \int_{0}^{T}dt\: X(t)$, we obtain
\begin{subequations}
    \begin{align}
    \overline{\scp{\psi_t}{P_\noneq\psi_t}}
    &= \sum_{e,e'\in\mathcal{E}} \overline{e^{i(e-e')t}} \, \scp{\psi_0}{\Pi_e P_\noneq \Pi_{e'}\psi_0}\\
    &= \sum_{e\in\mathcal{E}} \underbrace{\scp{\psi_0}{\Pi_e P_\noneq \Pi_e\psi_0}}_{\leq \varepsilon\delta \|\Pi_e \psi_0\|^2 \text{ by \eqref{ETH}}}\\
    &\leq \varepsilon\delta \sum_{e\in\mathcal{E}} \scp{\psi_0}{\Pi_e \psi_0}=\varepsilon\delta\,.
    \end{align}
\end{subequations}
Thus, for every $\eta>0$, there is $T_0>0$ such that for every $T>T_0$,
\be
\frac{1}{T}\int_0^T dt\: \scp{\psi_t}{P_\noneq\psi_t} < \varepsilon\delta + \eta\,.
\ee
By the Markov inequality,
\be
\frac{1}{T} \Bigl|\bigl\{t\in [0,T]: \scp{\psi_t}{P_\noneq\psi_t} >\varepsilon \bigr\} \Bigr| \leq \frac{\varepsilon\delta+\eta}{\varepsilon}= \delta + \frac{\eta}{\varepsilon}\,.
\ee
Taking the limit superior as $T\to\infty$, we find that
\be
\limsup_{T\to\infty}\frac{1}{T} \Bigl|\bigl\{t\in [0,T]: \scp{\psi_t}{P_\noneq\psi_t} >\varepsilon \bigr\} \Bigr| \leq \delta + \frac{\eta}{\varepsilon}\,.
\ee
Since $\eta>0$ was arbitrary, we must have that
\be
\limsup_{T\to\infty}\frac{1}{T} \Bigl|\bigl\{t\in [0,T]: \scp{\psi_t}{P_\noneq\psi_t} >\varepsilon \bigr\} \Bigr| \leq \delta
\ee
as claimed.

\subsection{Proof of Lemma~\ref{lea1}}

Since all the eigenvalues are non-negative, the maximal eigenvalue $\| M \|$ is bounded by $\tr M$, which is in turn bounded by $\varepsilon D$ by assumption.
To see that the bound is sharp, consider the matrix with all entries equal to $\varepsilon$, which has an eigenvector with all entries $1$ and eigenvalue $\varepsilon D$.

\subsection{Proof of Corollary~\ref{cor1}}

Corollary~\ref{cor1} is a consequence of Lemma~\ref{lea1}.
Let $P_{\noneq} \coloneqq I-P_\eq$.
Denote by $\Pi_e$ the orthogonal projection onto the eigenspace of $H_0$ corresponding to $e$.
In the basis $(\phi_k)$ restricted to the eigenspace corresponding to eigenvalue $e$, denoted by $(\phi_{e,j})_{j=1}^{D_e}$, the matrix corresponding to $\Pi_e P_{\noneq} \Pi_e$ has all diagonal entries $\leq \varepsilon$ by assumption.
Thus by Lemma~\ref{lea1}, we have
\be
    \| \Pi_e P_{\noneq} \Pi_e \| \leq \varepsilon D_e \leq \varepsilon D_{E}
\ee
for all $e$, and equivalently $H_0$ satisfies the ETH~\eqref{ETH} with tolerance $\varepsilon D_{E}$.

For the second part, let $\varepsilon' \coloneqq \min\{\varepsilon,\: 1/D_{E}\}$, let $E$ be an eigenvalue with degeneracy $D_E$, and $P_\noneq$ be the projection operator onto the one-dimensional subspace spanned by the vector
\be
    \psi \coloneqq \sqrt{\varepsilon'} \sum_{j=1}^{D_E} \phi_{E,j} + \sqrt{\eta} \sum_{e\: (\neq E)} \sum_{j=1}^{D_e} \phi_{e,j},
\ee
where $\eta \coloneqq \frac{1-\varepsilon' D_E}{D - D_E}$ to satisfy the normalization $\| \psi \| = 1$ (in the case $D=D_E$, note that $\varepsilon'=1/D_E$ by \eqref{eq:eps_lowerbound} and we only have the first term).
Here, $\varepsilon' \geq 1/D$ implies $\eta \leq 1/D \leq \varepsilon'$.
Then, since all diagonal elements of $P_\noneq$ in the basis $(\phi_{e,j})_{e,j}$ are either $\varepsilon'$ or $\eta$, the basis $(\phi_k)$ satisfies the condition~\eqref{ETHONB} for this $P_\noneq$.
Moreover, the matrix corresponding to $\Pi_E P_{\noneq} \Pi_E$ has all entries equal to $\varepsilon'$ in the basis $(\phi_{E,j})_{j=1}^{D_E}$.
Hence, we have $\| \Pi_E P_{\noneq} \Pi_E \| = \min\{\varepsilon D_{E},\: 1\}$ as claimed.

\subsection{Proof of Proposition~\ref{prop:mostphi}}

As shown by Reimann \cite[Eq.~(7)]{Reimann2015}, L\'evy's lemma implies that for a uniformly distributed unit vector $\phi$ in a $D$-dimensional Hilbert space $\Hilbert$ and a self-adjoint operator $A$ (which is not a multiple of the identity) and every $\eta\geq 0$,
\be\label{Reimann-bound}
\PPP\Bigl( \bigl| \langle\phi,A\phi\rangle -\tr(A)/D \bigr| \geq \eta \Bigr) \leq 2\exp \Bigl(-\frac{C \eta^2 D}{\Delta_A^2} \Bigr)
\ee
with $C$ as in Proposition~\ref{prop:mostphi} and $\Delta_A> 0$ the difference between the largest and the smallest eigenvalue of $A$. Now we specialize to our case with $\Hilbert=\Hilbert_e$, $D=D_e$, and $A= \Pi_e P_{\noneq} \Pi_e$ with $\Pi_e$ the projection to $\Hilbert_e$. Since $0\leq A \leq I$, we have that $\Delta_A\leq 1$ (so the right-hand side of \eqref{Reimann-bound} can only become larger if we replace $\Delta_A$ by 1). By assumption, there is a basis $(\phi_k)$ for which every vector lies in $\mathrm{MATE}_\varepsilon$; evaluating the trace in this basis, we find that $\tr(A)\leq \varepsilon D_e$. Thus,
\begin{subequations}\label{ConcentratinForEachState}
\begin{align}
    \mathbb{P}\Bigl(\phi\notin \textup{MATE}_{\varepsilon+\eta}\Bigr) &= \mathbb{P}\Bigl(\|P_{\mathrm{neq}}\phi\|^2 > \varepsilon+\eta\Bigr)\\
    &\leq \mathbb{P}\Bigl(\|P_{\mathrm{neq}}\phi\|^2 \geq \tr(A)/D_e + \eta\Bigr)\\
    &\leq \mathbb{P}\Bigl(\left|\|P_{\mathrm{neq}}\phi\|^2-\tr(A)/D_e\right|\geq \eta\Bigr)\\
    &\leq 2\exp\left(-C\eta^2 D_e\right).
\end{align}
\end{subequations}
by \eqref{Reimann-bound}, which completes the proof.

\begin{rmk}
    An alternative argument based on the Markov inequality instead of L\'evy's lemma (and the fact that the average of $\|P_{\noneq}\phi\|^2$ over $\SSS(\Hilbert_e)$ is $\leq\varepsilon$, still assuming that one eigen-ONB of $H_0$ satisfies \eqref{ETHONB}) yields for any $\delta>0$ that $(1-\delta)$-most $\psi\in\SSS(\Hilbert_e)$ lie in $\text{MATE}_{\varepsilon/\delta}$.
\end{rmk}

\subsection{Proof of Proposition~\ref{prop: most eigenbases}}

Since each eigenvector $\phi_k$ in the random eigen-ONB is uniformly distributed in a $D_e$-dimensional Hilbert space $\Hilbert_e$, the inequality~\eqref{ConcentratinForEachState} applied to each $\phi_{k}$ yields
\begin{subequations}
\begin{align}
    \mathbb{P}\Bigl(\exists k\leq D_e: \phi_k\notin \textup{MATE}_{\varepsilon+\eta}\Bigr) &\leq \sum_{k=1}^{D_e}\mathbb{P}\Bigl(\phi_k\notin \textup{MATE}_{\varepsilon+\eta}\Bigr)\\
    &\leq 2D_e \exp\left(-C\eta^2 D_e\right),
\end{align}
\end{subequations}
which completes the proof.

\subsection{Proof of Proposition~\ref{prop: most eigenbases H0}}
For each eigenvalue $e$ of $H_0$ let $\Hilbert_e$ denote the corresponding eigenspace and $D_e$ its dimension.
Note that
    \begin{align}\label{eq: p all eigenvectors mate}
        \mathbb{P}\Bigl(\exists k\leq D: \phi_k\notin \textup{MATE}_{\varepsilon+\eta}\Bigr) \leq \sum_{e\in E} \mathbb{P}\Bigl(\exists k\leq D: \phi_k\in \Hilbert_e\ \mathrm{and}\ \phi_k\notin \textup{MATE}_{\varepsilon+\eta}\Bigr).
    \end{align}
For all eigenvalues $e$ with $D_e\leq \eta/\eps$, the corresponding eigenvectors lie in $\mathrm{MATE}_{\varepsilon+\eta}$ by Corollary~\ref{cor1}.
Hence the corresponding summands in \eqref{eq: p all eigenvectors mate} are zero.
On the other hand, if $D_e\geq \eta/\eps$, Proposition~\ref{prop: most eigenbases} implies that
    \begin{align}
        \mathbb{P}\Bigl(\exists k\leq D: \phi_k\in \Hilbert_e\ \mathrm{and}\ \phi_k\notin \textup{MATE}_{\varepsilon+\eta}\Bigr)\leq 2D_e \exp\left(-C \eta^2 D_e\right)
    \end{align}
Since the function $x\to x e^{-x}$ is decreasing for $x\geq 1$, by the assumption on $\eta$ that $C \eta^3 \geq \varepsilon$, the right hand side is bounded from above by $2\frac{\eta}{\eps} \exp\left(-C \eta^3 /\eps\right)$.
Furthermore, since the number of eigenspaces of dimension larger than $\eta/\eps$ is at most $D \eps/\eta$, we obtain
the bound
    \begin{align}
        \mathbb{P}\Bigl(\exists k\leq D: \phi_k\notin \textup{MATE}_{\varepsilon+\eta}\Bigr) \leq 2D \exp\left(-C \eta^3/\varepsilon\right).
    \end{align}

\subsection{Proof of Lemma~\ref{lem:1}}\label{sec:pflem1}

We cite a key fact from Appendix~A in \cite{TTV22-mathe}:

\begin{lemma}\label{lem:TTV22-mathe}
    If $H$ has continuous distribution in the Hermitian $n\times n$ matrices, then with probability 1 it has non-degenerate eigenvalues and eigenvalue gaps.
\end{lemma}

Lemma~\ref{lem:1} says more in that there is a whole interval $(0,\lambda_0)$ of $\lambda$ values for which $H_0+\lambda V$ will have non-degenerate eigenvalues and eigenvalue gaps. As a preparation for the proof, we
establish the following lemma:

\begin{lemma}\label{lem: polynomials}
    The set of Hermitian $n\times n$ matrices with degenerate eigenvalues can be written as the zero set of a polynomial in the matrix entries. Likewise, the set of Hermitian $n\times n$ matrices with distinct eigenvalues but with degenerate eigenvalue gaps can be written as the zero set of a polynomial in the matrix entries. 
\end{lemma}

\begin{proof}
    Let $A$ be a Hermitian $n\times n$ matrix with eigenvalues $\lambda_1,\dots,\lambda_n$. The matrix $A$ has degenerate eigenvalues if and only if its discriminant
    \begin{align}
        \mbox{disc}(A) = \prod_{i<j} (\lambda_i-\lambda_j)^2 
    \end{align}
    vanishes. Since the discriminant of a matrix can be written as a polynomial in the matrix entries, see, e.g., Lemma~1 in \cite{Parlett02}, the first claim follows. 

    For the second claim, we follow the proof strategy of Lemma~1 in \cite{Parlett02} and adapt it to our situation. Let $A$ be a Hermitian $n\times n$ matrix with distinct eigenvalues $\lambda_1,\dots,\lambda_n$. Then $A$ has degenerate eigenvalue gaps if and only if
    \begin{align}
        \prod_{(i,j,k,l)\in I} \left(\lambda_i-\lambda_j-(\lambda_k-\lambda_l)\right) = 0,\label{eq: prod gaps}
    \end{align}
    where
    \begin{align}
        I \coloneqq  \left\{(i,j,k,l) \in [n]^4: (i\neq k\mbox{ or } j\neq l) \mbox{ and } (i\neq j \mbox{ or } k\neq l)\right\}
    \end{align}
    and $[n]\coloneqq \{1,\dots,n\}$. Since the tuples $(i,j,k,l) \in I$ with $i=j$ and $k\neq l$ (or $k=l$ and $i\neq j$) lead to non-zero factors in \eqref{eq: prod gaps} (due to the non-degeneracy of the eigenvalues), we can replace the set $I$ in \eqref{eq: prod gaps} by the set
    \begin{align}
        I'\coloneqq  \left\{(i,j,k,l) \in [n]^4: (i\neq k\mbox{ or } j\neq l) \mbox{ and } (i\neq j \mbox{ and } k\neq l)\right\}.
    \end{align}
    We enumerate the eigenvalue differences $(\lambda_i-\lambda_j)_{i\neq j}$ as $x_1\coloneqq \lambda_1-\lambda_2, x_2\coloneqq \lambda_1-\lambda_3, \dots, x_{n-1} \coloneqq \lambda_1-\lambda_n, x_n\coloneqq \lambda_2-\lambda_1, x_{n+1}\coloneqq \lambda_2-\lambda_3,\dots, x_M \coloneqq  \lambda_n-\lambda_{n-1}$, where $M\coloneqq n(n-1)$, and consider the Vandermonde matrix $V=V(x_1,\dots,x_M)$ which is defined as
    \begin{align}
        V(x_1,\dots,x_M) = \left(\begin{matrix}
            1 & x_1 & x_1^2 & \dots & x_1^{M-1}\\
            1 & x_2 & x_2^2 & \dots & x_2^{M-1}\\
            1 & x_3 & x_3^2 & \dots & x_3^{M-1}\\
            \vdots & \vdots & \vdots & \ddots & \vdots \\
            1 & x_M & x_M^2 & \dots & x_M^{M-1}
        \end{matrix}\right).
    \end{align}
    It is well known that
    \begin{align}
        \det V = \prod_{1\leq i < j \leq M} (x_j-x_i)
    \end{align}
    and thus
\begin{subequations}
    \begin{align}
        \left(\det V\right)^2 &= \prod_{1\leq i<j\leq M} (x_j-x_i)^2 = (-1)^{M(M-1)/2} \prod_{i\neq j} (x_i-x_j)\\
        &= (-1)^{M(M-1)/2} \prod_{(i,j,k,l)\in I'} (\lambda_i-\lambda_j-(\lambda_k-\lambda_l)).
    \end{align}
\end{subequations}
Therefore $A$ has degenerate eigenvalue gaps if and only if $(\det V)^2=0$.

    We define the $M\times M$ matrix $B=(B_{ij})$ in the following way:
    \begin{align}
        B_{ij} \coloneqq  \sum_{k=1}^{M} x_k^{i+j-2}.
    \end{align}
    One immediately sees that $B=V^T V$ which implies $\det B = (\det V)^2$. Obviously, $B_{11}=M$ and for $(i,j)\neq (1,1)$ we have that
\begin{subequations}
    \begin{align}
        B_{ij} &= \sum_{k,l=1}^n (\lambda_k-\lambda_l)^{i+j-2}\\
        &= \sum_{k,l=1}^n \sum_{p=0}^{i+j-2} \binom{i+j-2}{p} \lambda_k^{i+j-2-p} \lambda_l^p (-1)^p\\
        &= \sum_{p=0}^{i+j-2} \binom{i+j-2}{p} (-1)^p \tr(A^{i+j-2-p}) \tr(A^p).
    \end{align}
\end{subequations}
    We conclude that the entries of $B$ are polynomials in the entries of $A$ and thus that also $\det B=(\det V)^2$ is a polynomial in the entries of $A$. This proves the second claim.
\end{proof}

\begin{lemma}\label{lem:fixedV}
    Suppose there is $\tilde\lambda>0$ such that $H_0+\tilde\lambda V$ has non-degenerate eigenvalues and eigenvalue gaps. Then there is $\lambda_0>0$ such that for every $\lambda\in(0,\lambda_0)$, $H=H_0+\lambda V$ has non-degenerate eigenvalues and eigenvalue gaps.
\end{lemma}

\begin{proof}
    We consider $H$ as a function of $\lambda$. By Lemma~\ref{lem: polynomials} there is a polynomial $P_1$ in the entries of $H$ and therefore a polynomial $\tilde{P}_1$ in $\lambda$ such that its zeros are exactly the matrices with degenerate eigenvalues. By assumption, $\tilde{P}_1(\tilde{\lambda})\neq 0$. Thus, $\tilde{P}_1$ does not vanish identically and therefore $\tilde{P}_1$ has only finitely many zeros. We know that $\lambda=0$ is one of the zeros and therefore there exists a $\tilde{\lambda}_0$ such that $H=H_0+\lambda V$ has non-degenerate eigenvalues for all $\lambda \in (0,\tilde{\lambda}_0)$.\footnote{\label{fn:alg}Here is an alternative proof of the eigenvalue statement of Lemma~\ref{lem:fixedV}: Consider the polynomial $P(\lambda,E)=\det(H_0+\lambda V- E I)$ in 2 variables, which vanishes if and only if $E$ is an eigenvalue of $H_0+\lambda V$. $P$ has degree $\leq D$ because $\det$ is a degree-$D$ polynomial in the matrix entries, and each entry is a degree-1 polynomial in $(\lambda,E)$. Since for $\lambda=\tilde\lambda$, $P$ has $D$ distinct zeros, $P$ has degree $D$ and is square-free. Then the number of singular points in the plane is finite (e.g., \cite[Bemerkung 3.2]{Fis94}). The zero set $S$ of $P$ is known to consist of finitely many smooth curves that are either closed or tend to infinity in both directions, and can intersect themselves or each other only in singular points. There are only finitely many points $p$ on the smooth curves where the tangent is vertical, i.e., parallel to the $E$ axis, because at such points $p$, $\partial P/\partial E(p)=0$, so $p$ is a joint zero of $P$ and $\partial P/\partial E$; since these two polynomials have no common prime factor (see the proof of \cite[Bemerkung 3.2]{Fis94}), they intersect only in finitely many points by B\'ezout's theorem. Now along any vertical line $L$, points $p$ in $L\cap S$ are simple roots unless either $p$ is a singular point or a curve has a vertical tangent at $p$. Thus, for every $\lambda$ except finitely many exceptions, every zero of $P$ in $E$ is simple, so $H$ is non-degenerate. Let $\lambda_0$ be the smallest positive exception.}

    Again it follows from Lemma~\ref{lem: polynomials} that there is a polynomial $P_2$ in the entries of $H$ and therefore a polynomial $\tilde{P}_2$ in $\lambda \in (0,\tilde{\lambda}_0)$ such that its zeros are exactly the matrices with degenerate eigenvalue gaps (but non-degenerate eigenvalues). Note that $\tilde{P}_2$ can be considered as a polynomial on $[0,\infty)$ that vanishes in the (finitely many) zeros of $\tilde{P}_1$. Because of $\tilde{P}_2(\tilde{\lambda})\neq 0$, $\tilde{P}_2$ does not vanish identically and has therefore only finitely many zeros. Its zeros in $(0,\tilde{\lambda}_0)$ are matrices with non-degenerate eigenvalues but degenerate eigenvalue gaps. Since there are only finitely many such zeros, it follows that there exists a $0<\lambda_0\leq \tilde{\lambda}_0$ such that $H=H_0+\lambda V$ has non-degenerate eigenvalues and non-degenerate eigenvalue gaps for all $\lambda \in (0,\lambda_0)$. 
\end{proof}

Now Lemma~\ref{lem:1} follows from Lemma~\ref{lem:TTV22-mathe} and Lemma~\ref{lem:fixedV} by fixing $\tilde\lambda=1$ and choosing a $V$ for which $H_0+V$ has non-degenerate eigenvalues and eigenvalue gaps.

\subsection{Proof of Theorem~\ref{thm1 new}}
By Lemma~\ref{lem:1}, the probability $\mathbb{P}_V(H_0+\lambda V \ \text{is non-degenerate for all}\ 0<\lambda<\lambda_1)$ over the distribution of $V$ tends to 1 as $\lambda_1\to 0$.
We choose $\lambda_1$ such that this probability is larger than $(1-\alpha)^{1/2}$.
Let $S$ denote the set of all $V$ such that $H_0+\lambda V$ is non-degenerate for all $0<\lambda <\lambda_1$.
By Proposition~\ref{prop: most eigenbases H0}, $(1-\beta)$-most normalized eigen-ONB's of $H_0$ lie in $\mathrm{MATE}_{\varepsilon+\eta}$ with $\beta = 2D\exp( -C \eta^3 / \varepsilon )$.
For $V\in S$ consider the probability  $\mathbb{P}_U(H_0+\lambda  U V U^\dagger \ \text{satisfies } \mathrm{ETH}_{\varepsilon+2\eta} \ \text{for all }\ 0<\lambda<\lambda_0)$ over the unitaries $U$ commuting with $H_0$.
For $\lambda_0\to 0$ this probability becomes $\geq (1-\beta)$.
Since the distribution of $V$ is invariant under these unitaries, and since the spectrum is invariant under conjugation, we obtain
\begin{multline}
    \mathbb{P}_V(H_0+\lambda V\ \text{satisfies}\ \text{ETH}_{\varepsilon+2\eta}\ \text{for all}\ 0<\lambda<\lambda_0| V \in S)\\
    =\mathbb{E}_U \mathbb{P}_V(H_0+\lambda UVU^\dagger\ \text{satisfies}\ \text{ETH}_{\varepsilon+2\eta}\ \text{for all}\ 0<\lambda<\lambda_0|V \in S)\\
    =\mathbb{E}_V\Bigl(\mathbb{P}_U(H_0+\lambda UVU^\dagger\ \text{satisfies}\ \text{ETH}_{\varepsilon+2\eta}\ \text{for all}\ 0<\lambda<\lambda_0)|V \in S\Bigr)
\end{multline}
which tends to something $\geq (1-\beta)$ for $\lambda_0\to 0$.
Hence, there is a $0<\lambda_0\leq \lambda_1$ such that for $(1-\alpha)^{1/2}(1-\beta)$-most $V$ in $S$ the eigenbasis of $H_0+\lambda  V  $ satisfies ETH$_{\varepsilon+2\eta}$ for all $0<\lambda<\lambda_0$.
Combining this with the probability that $V$ lies in $S$, we conclude that for $(1-\alpha)(1-\beta)$-most $V$ the perturbed Hamiltonian satisfies $\text{ETH}_{\varepsilon+2\eta}$ for all $0<\lambda<\lambda_0$.
The statement about thermalization then follows from Proposition~\ref{prop:ETH}.

\subsection{Proof of Theorem~\ref{thm2}}

In dimension one, the eigenstates $\Psi_k$ defined in \eqref{eq: energy eigenstate highdim} have energies $E_k \coloneqq  -2\sum_{i=1}^N \cos k_i$.
We first remark that the assumption that $L$ is a prime number ensures that all degeneracies in the spectrum of $\HfF$ are trivial, i.e., only due to changing the signs of the $k_i$.  This is shown in \cite{ST23} at the beginning of the proof of Theorem~3.2. Note that the model considered there agrees with the model in the present paper in the case that $\theta=0$. (This parameter is introduced in \cite{ST23} to remove the degeneracies and to this end has to be chosen to be small but non-zero.)

The proof of Theorem~\ref{thm2} will make use of several propositions that we formulate now and prove in the subsequent subsections. The first one states that the expectation of $N_\Gamma$ in an arbitrary eigenstate of $\HfF$ is close to $N|\Gamma|/|\Lambda|$, provided that $N$ is sufficiently large. 

\begin{prop}[Expectation of $N_\Gamma$ in arbitrary eigenstates] \label{prop: exp NGamma es} 
Let $d=1$ and let $k=(k_1,\dots,k_N) \in \tilde{\mathcal{K}}^N$. Let $E_k \coloneqq  -2\sum_{i=1}^N \cos k_i$ be the corresponding eigenvalue and $\Hilbert_{E_k} = \mathrm{span}\{\Psi_{k'}: k'_j = \pm k_j \mbox{\; for all\; }j\}$ the corresponding eigenspace. Let $\Gamma \subset\Lambda$ be an interval, and let $\phi \in \mathbb{S}(\Hilbert_{E_{k}})$. Then
	\begin{align}
		\Biggl|\langle\phi,N_\Gamma\phi\rangle - N\frac{|\Gamma|}{|\Lambda|} \Biggr|	\leq \ln N+1.
	\end{align}	
\end{prop}

For the proof of Proposition~\ref{prop: exp NGamma es} we need the following proposition concerning the expectation of $N_\Gamma$ in the eigenstates $\Psi_k$ of $\HfF$. 

\begin{prop}[Expectation of $N_\Gamma$ in eigenstates $\Psi_k$]\label{prop: exp NGamma Psik}
	Let $d\geq 1$, let $k, k' \in \tilde{\mathcal{K}}^N$ and $x\in \Gamma$. Then
\begin{align}
\langle \Psi_k,N_\Gamma \Psi_k\rangle = N \frac{|\Gamma|}{|\Lambda|}
\end{align}
and
\begin{align}
    \langle \Psi_{k},c_x^\dagger c_x\Psi_{k'}\rangle = \frac{1}{|\Lambda|} \mathrm{sgn}(\tilde{\sigma}) e^{i(k'_{\tilde{\sigma}^{-1}(l)}-k_l)\cdot x}
\end{align}
if only $k_l$ does not appear in $k'$ and $\tilde{\sigma}\in\mathcal{S}_N$ is the permutation such that $k'_{\tilde{\sigma}^{-1}(m)} = k_m$ for all $m\neq l$.
Thus, if  exactly one component of $k$ does not appear in $k'$, then 
	\begin{align}
\left|\langle\Psi_{k},N_\Gamma\Psi_{k'}\rangle 
  \right| \leq \frac{|\Gamma|}{|\Lambda|}.
	\end{align}
If more than one component of $k$ does not appear in $k'$, then \be \langle\Psi_{k},N_\Gamma\Psi_{k'}\rangle = 0\,.
\ee
\end{prop}

For the proof of Proposition~\ref{prop: exp NGamma Psik} we need the following lemma:

\begin{lemma}\label{lem: comb formula}
	Let $d\geq 1$, $x_1,\dots,x_N \in \{1,\dots,L\}$, and let $k \in \mathcal{K}^N$ such that $\tau(k)\coloneqq (k_{\tau(1)},\dots,k_{\tau(N)})\in \tilde{\mathcal{K}}^N$ for some $\tau\in\mathcal{S}_N$, where $\mathcal{S}_N$ denotes the symmetric group. Then
	\begin{align}
		\langle \Phi_{\textup{vac}},c_{x_N}\dots c_{x_1} a_{k_1}^\dagger \dots a_{k_N}^\dagger\Phi_{\textup{vac}}\rangle = \frac{1}{L^{Nd/2}} \sum_{\sigma \in \mathcal{S}_N} \mathrm{sgn}(\sigma) \prod_{j=1}^N e^{ik_j\cdot x_{\sigma(j)}}. \label{eq: prod c_x a_k}
	\end{align}
\end{lemma}

This formula is well known; it was stated, e.g., in \cite{ST23} in (C3) in the case that $d=1$. For convenience of the reader, we include the proof in the next subsection.

The next proposition shows that the variance of $N_\Gamma$ in an arbitrary eigenstate of $\HfF$ is small provided that $N$ is sufficiently large. 

\begin{prop}[Variance of $N_\Gamma$ in arbitrary eigenstates]\label{prop: var NGamma arb es} Let $d=1$ and let $k\in \tilde{\mathcal{K}}^N$, $\mathcal{H}_{E_k}$, $\Gamma$ and $\phi$ be as in Proposition~\ref{prop: exp NGamma es}. Then,
	\begin{align}
	\langle\phi,N_\Gamma^2\phi\rangle  - \left(\langle\phi,N_\Gamma\phi\rangle\right)^2 &\leq\; 4 N \ln N + 13 N + 3(\ln N)^2 + 13 \ln N + 10\\
    &\stackrel{N\geq 46}{\leq} \;8 N\ln N\,.
	\end{align}
\end{prop}
Before proving Proposition~\ref{prop: var NGamma arb es}, we show a similar statement for the variance of $N_\Gamma$ in the eigenstates $\Psi_k$ of $\HfF$.

\begin{prop}[Variance of $N_\Gamma$ in eigenstates $\Psi_k$]\label{prop: var NGamma Psik} 
Let $d\geq 1$ and $k\in\tilde{\mathcal{K}}^N$. Then
\begin{align}
    \langle\Psi_k,N_\Gamma^2\Psi_k\rangle - \left(\langle\Psi_k,N_\Gamma\Psi_k\rangle\right)^2 \leq N \frac{|\Gamma|}{|\Lambda|}\left(1-\frac{|\Gamma|}{|\Lambda|}\right).
\end{align}
\end{prop}

\begin{proof}[Proof of Theorem~\ref{thm2}]
    Let $\phi\in\SSS(\Hilbert)$ be any eigenvector of $\HfF$. The Born distribution $\PPP$ associated with $\phi$ and the observable $N_\Gamma$ has expectation $E\coloneqq \langle \phi, N_\Gamma \phi\rangle$ and variance $V\coloneqq \langle \phi, N_\Gamma^2 \phi\rangle - \langle\phi, N_\Gamma\phi\rangle^2$. Writing $\overline{B}_r(x)$ for $[x-r,x+r]$, we can express Chebyshev's inequality as
    \be
    \PPP\Bigl(\overline{B}_{\alpha\sqrt{V}}(E)\Bigr)\geq 1-\frac{1}{\alpha^2} 
    \ee
    for any $\alpha>0$. By Propositions~\ref{prop: exp NGamma es} and \ref{prop: var NGamma arb es} for $N\geq 46$,
    \be
    \overline{B}_{\alpha\sqrt{V}}(E)
    \subseteq \overline{B}_{\ln N + 1+\alpha \sqrt{8N \ln N}}\biggl(N \frac{|\Gamma|}{|\Lambda|}\biggr)
    \subseteq \overline{B}_{N\eta}\biggl(N \frac{|\Gamma|}{|\Lambda|}\biggr)
    \ee
    for $\alpha=\frac{N\eta}{2\sqrt{8N\ln N}}$ using $N\eta/2>\ln N +1$, so
    \be
    \| P_{\eq,\eta} \phi \|^2 = \PPP\Biggl(\overline{B}_{N\eta}\biggl(N \frac{|\Gamma|}{|\Lambda|}\biggr)\Biggr)\geq 1-\frac{32 N \ln N}{N^2\eta^2}\,, 
    \ee
    which is equivalent to \eqref{eqthm2}.
\end{proof}

\subsection{Proof of Lemma~\ref{lem: comb formula}}

Without loss of generality assume that $k\in\tilde{\mathcal{K}}^N$. We prove \eqref{eq: prod c_x a_k} by induction. First note that as a consequence of the canonical anticommutation relations we immediately see from the definition of the $a_k^\dagger$ that $\{c_x,a_k^\dagger \} = e^{ikx}/ L^{d/2}$ and $\{a_k^\dagger, a_{k'}^\dagger \}=0$. Now \eqref{eq: prod c_x a_k} can be shown by induction. For $N=1$ the equation holds because 
\begin{align}
	\langle \Phi_{\textup{vac}}, c_{x_1} a_{k_1}^\dagger \Phi_{\textup{vac}}\rangle = \frac{e^{ik_1\cdot
			x_1}}{ L^{d/2} } \langle\Phi_{\textup{vac}},\Phi_{\textup{vac}}\rangle - \langle\Phi_{\textup{vac}}, a_{k_1}^\dagger c_{x_1}\Phi_{\textup{vac}}\rangle = \frac{e^{ik_1 \cdot x_1}}{ L^{d/2} }.
\end{align}
Now suppose that \eqref{eq: prod c_x a_k} holds for some $N\in\mathbb{N}$. Then we have that
\begin{subequations}
	\begin{align}
        &\langle \Phi_{\textup{vac}}, c_{x_{N+1}} \dots c_{x_1} a_{k_1}^\dagger \dots a_{k_{N+1}}^\dagger\Phi_{\textup{vac}}\rangle \nonumber\\
        &= (-1)^N \langle\Phi_{\textup{vac}}, c_{x_{N+1}} \dots c_{x_1} a_{k_{N+1}}^\dagger a_{k_1}^\dagger \dots a_{k_{N}}^\dagger\Phi_{\textup{vac}}\rangle\\
        &= (-1)^N \left(\frac{e^{ik_{N+1}\cdot
				x_1}}{ L^{d/2} }\langle\Phi_{\textup{vac}},c_{x_{N+1}} \dots c_{x_2} a_{k_1}^\dagger \dots a_{k_N}^\dagger\Phi_{\textup{vac}}\rangle \right.\nonumber\\
        &\specialcell{\hfill  \left. - \langle\Phi_{\textup{vac}},c_{x_{N+1}} \dots c_{x_2} a_{k_{N+1}}^\dagger c_{x_1}a_{k_1}^\dagger \dots a_{k_N}^\dagger \Phi_{\textup{vac}}\rangle \right)}\\
		&= (-1)^N \left(\frac{e^{ik_{N+1} \cdot x_1}}{ L^{d/2} }\langle\Phi_{\textup{vac}},c_{x_{N+1}}\dots c_{x_2} a_{k_1}^\dagger \dots a_{k_N}^\dagger\Phi_{\textup{vac}}\rangle \right. \nonumber\\
        & \specialcell{\hfill - \frac{e^{ik_{N+1}\cdot x_2}}{ L^{d/2} } \langle\Phi_{\textup{vac}},c_{x_{N+1}}\dots c_{x_3}c_{x_1} a_{k_1}^\dagger \dots a_{k_N}^\dagger\Phi_{\textup{vac}}\rangle \hfill} \nonumber\\
		&\specialcell{\hfill \left. + \langle\Phi_{\textup{vac}},c_{x_{N+1}}\dots c_{x_3} a_{k_{N+1}}^\dagger c_{x_2} c_{x_1} a_{k_1}^\dagger \dots a_{k_N}^\dagger\Phi_{\textup{vac}}\rangle\right) }\\
		&=\cdots\\
		&= (-1)^N \sum_{l=1}^{N+1} (-1)^{l+1} \frac{ e^{ik_{N+1}\cdot x_l} }{ L^{d/2} } \langle\Phi_{\textup{vac}},c_{x_{N+1}}\dots c_{x_{l+1}} c_{x_{l-1}}\dots c_{x_1} a_{k_1}^\dagger \dots a_{k_N}^\dagger\Phi_{\textup{vac}}\rangle\\
		&= \frac{1}{ L^{(N+1)d/2} } \sum_{l=1}^{N+1} (-1)^{N+l+1} e^{ik_{N+1}\cdot x_l} \sum_{\sigma\in \mathcal{S}_{N,l}} \mbox{sgn}(\sigma) \prod_{j=1}^N e^{ik_j \cdot x_{\sigma(j)}},
	\end{align}
\end{subequations}
where $\mathcal{S}_{N,l}$ denotes the set of permutations $\sigma: \{1,\dots,N\} \to \{1,\dots,l-1,l+1,\dots,N+1\}$ of the set $\{1,\dots,l-1,l+1,\dots,N+1\}$. Note that we used the induction hypothesis in the last step.

Any $\sigma \in \mathcal{S}_{N,l}$ is related to a permutation $\tau \in \mathcal{S}_{N+1}$ with $\tau(N+1)=l$ via $N+1-l$ transpositions and vice versa. Therefore, we obtain
\begin{subequations}
	\begin{align}
		&\langle \Phi_{\textup{vac}}, c_{x_{N+1}} \dots c_{x_1} a_{k_1}^\dagger \dots a_{k_{N+1}}^\dagger\Phi_{\textup{vac}}\rangle \nonumber\\ 
        &= \frac{1}{ L^{(N+1)d/2} } \sum_{l=1}^{N+1} (-1)^{N+1-l} \sum_{\substack{\tau\in\mathcal{S}_{N+1}\\ \tau(N+1)=l}} (-1)^{N+l+1} \mbox{sgn}(\tau) \prod_{j=1}^{N+1} e^{ik_j\cdot x_{\tau(j)}}\\
		&= \frac{1}{ L^{(N+1)d/2} } \sum_{\tau\in\mathcal{S}_{N+1}} \mbox{sgn}(\tau) \prod_{j=1}^{N+1} e^{ik_j\cdot x_{\tau(j)}},
	\end{align}
\end{subequations}
which finishes the proof of \eqref{eq: prod c_x a_k}. 

\subsection{Proof of Proposition~\ref{prop: exp NGamma Psik}}
Let $x\in\Gamma$. With the help of Lemma~\ref{lem: comb formula} we find that
\begin{subequations}
\begin{align}
	&\langle\Psi_{k},c_x^\dagger c_x \Psi_{k'}\rangle \nonumber\\ 
    &= \langle\Phi_{\textup{vac}},a_{k_N}\dots a_{k_1} c_x^\dagger c_x a_{k_1'}^\dagger \dots a_{k_N'}^\dagger \Phi_{\textup{vac}}\rangle\\
	&= \frac{1}{L^{Nd/2}} \sum_{x_1,\dots,x_N\in\Lambda} e^{-ik_1\cdot x_1} \dots e^{-ik_N\cdot x_N} \langle\Phi_{\textup{vac}},c_{x_N} \dots c_{x_1} c_x^\dagger c_x a_{k_1'}^\dagger \dots a_{k_N'}^\dagger\Phi_{\textup{vac}}\rangle\\
	&= \frac{1}{L^{Nd/2}} \sum_{x_1,\dots,x_N\in\Lambda} e^{-ik_1\cdot x_1} \dots e^{-ik_N\cdot x_N} \chi_{\{x \in \{x_1,\dots,x_N\} \}} \langle\Phi_{\textup{vac}},c_{x_N}\dots c_{x_1} a_{k_1'}^\dagger \dots a_{k_N'}^\dagger\Phi_{\textup{vac}}\rangle\\
    &= \frac{1}{L^{Nd}} \sum_{x_1,\dots,x_N\in\Lambda} e^{-ik_1\cdot x_1}\dots e^{-ik_N\cdot x_N} \chi_{\{x \in \{x_1,\dots,x_N\} \}} \sum_{\sigma\in\mathcal{S}_N} \mbox{sgn}(\sigma) \prod_{j=1}^{N} e^{ik_j'\cdot x_{\sigma(j)}}\\
	&= \frac{1}{L^{Nd}} \sum_{l=1}^N e^{-ik_l\cdot x} \sum_{\substack{x_1,\dots,x_{l-1},x_{l+1},\dots,x_N\in\Lambda,\\ x_l=x}} \left(\prod_{\substack{m=1\\m\neq l}}^N e^{-ik_m\cdot x_m}\right) \sum_{\sigma \in \mathcal{S}_N} \mbox{sgn}(\sigma) \prod_{j=1}^N e^{ik_j'\cdot x_{\sigma(j)}}\\
	&= \frac{1}{L^{Nd}} \sum_{\sigma \in \mathcal{S}_N} \mbox{sgn}(\sigma) \sum_{l=1}^N e^{i(k_{\sigma^{-1}(l)}'-k_l)\cdot x} \prod_{\substack{m=1\\m\neq l}}^N \left(\sum_{x_m\in\Lambda} e^{i(k_{\sigma^{-1}(m)}'-k_m)\cdot x_m}\right).\label{eq: cx^t cx}
\end{align}
\end{subequations}
For $k,k' \in \mathcal{K}$ we have
\begin{subequations}
	\begin{align}
		\sum_{y \in \Lambda} e^{i(k'-k)\cdot y} &= \sum_{y_1,\dots,y_d=1}^L e^{i(k_{1}'-k_{1})y_1} \cdots e^{i(k_{d}'-k_{d})y_d}\\
		&= \prod_{m=1}^d \sum_{y_m=1}^L e^{i(k_{m}'-k_{m})y_m}\\
		&= \prod_{m=1}^{d} \left(L\delta_{k_m' k_m} + \chi_{\{k_m'\neq k_m\}} \frac{e^{i(k_m'-k_m)}-e^{i(k_m'-k_m)(L+1)}}{1-e^{i(k_m'-k_m)}}\right)\\
		&= L^d \delta_{k' k},\label{eq: sum exp higherdim},
	\end{align}
\end{subequations}
where we used that $(k_m'-k_m)L$ is a multiple of $2\pi$ and therefore $e^{i(k_m'-k_m)L}=1$.

Thus we see that if $k=k' \in \tilde{\mathcal{K}}^{N}$, only the permutation $\sigma =\mbox{id}$ gives a non-vanishing contribution in \eqref{eq: cx^t cx} and we obtain
\begin{align}
	\langle\Psi_k, c_x^\dagger c_x \Psi_k\rangle = \frac{N}{L^d}\label{eq: exp cx^t cx}
\end{align}
independently of $x\in \Gamma$. This implies
\begin{align}
	\langle\Psi_k,N_\Gamma\Psi_k\rangle = \frac{|\Gamma| N}{L^d}. \label{eq: k NGamma k}
\end{align}

If $k\neq k'$, then \eqref{eq: cx^t cx} only does not vanish if exactly one component of $k$ and $k'$ is different. Assume that $k_l$ for some $1\leq l \leq N$ does not appear in $k'$ and let $\tilde{\sigma}\in\mathcal{S}_N$ be the permutation such that $k'_{\tilde{\sigma}^{-1}(m)} = k_m$ for all $m\neq l$. Then we get
\begin{align}
	\langle\Psi_{k},c_x^\dagger c_x\Psi_{k'}\rangle = \frac{1}{L^d} \mbox{sgn}(\tilde{\sigma}) e^{i(k'_{\tilde{\sigma}^{-1}(l)}-k_l)\cdot x}\label{eq: Psik' cx^t cx Psik}
\end{align}
and therefore
\begin{align}
	\Bigl|\langle\Psi_{k},N_\Gamma \Psi_{k'}\rangle \Bigr| \leq \frac{|\Gamma|}{L^d}.\label{eq: k' NGamma k}
\end{align}
Combining \eqref{eq: k NGamma k} and \eqref{eq: k' NGamma k} and using that $|\Lambda| = L^d$ finishes the proof.

\subsection{Proof of Proposition~\ref{prop: exp NGamma es}}
Since $\HfF$ is invariant under cyclic permutations of $\Lambda$, there is no  loss of generality in assuming $\Gamma=\{1,\ldots,|\Gamma|\}$.

We first consider the case that there are no $k_l,k_m$ such that $k_l=-k_m$. In this case, we can assume without loss of generality that $k_j\geq 0$ for all $j$ as the energy is invariant under flipping the sign of any $k_j$. 
We express $\phi$ in the basis of the $\Psi_{k'}$ with $k'_j = \pm k_j$ for all $j$, i.e., we write
\begin{align}
	|\phi\rangle = \sum_{k'} \alpha_{k'} |\Psi_{k'}\rangle, \label{eq: phi}
\end{align}
where $\alpha_{k'} = \langle\Psi_{k'}|\phi\rangle$. 
We compute
\begin{subequations}
\begin{align}
	\langle \phi,N_\Gamma \phi\rangle &= \sum_{k',k''} \alpha_{k'}^* \alpha_{k''} \langle\Psi_{k'},N_\Gamma \Psi_{k''}\rangle\\
	&= \sum_{k'} |\alpha_{k'}|^2 N \frac{|\Gamma|}{|\Lambda|} + \sum_{k'\neq k''} \alpha_{k'}^* \alpha_{k''} \langle\Psi_{k'},N_\Gamma\Psi_{k''}\rangle\\
	&= N\frac{|\Gamma|}{|\Lambda|} + \sum_{x=1}^{|\Gamma|}\sum_{k'\neq k''} \alpha_{k'}^* \alpha_{k''} \langle\Psi_{k'},c_x^\dagger c_x \Psi_{k''}\rangle.
\end{align}
\end{subequations}
Because of Proposition~\ref{prop: exp NGamma Psik} we see that  $\langle\Psi_{k'},c_x^\dagger c_x\Psi_{k''}\rangle$ with $k'\neq k''$ does not vanish only if $k'$ and $k''$ differ in exactly one component. First suppose that $0<k_j<\pi$ for all $j$. For each $j$, let $k^{\prime (j)}$ denote the vector obtained from $k'$ by flipping the sign of its $j$th component. Then, by Proposition~\ref{prop: exp NGamma Psik} we find that
\begin{subequations}
\begin{align}
\sum_{x=1}^{|\Gamma|}\sum_{k'\neq k''} &\alpha_{k'}^* \alpha_{k''} \langle\Psi_{k'},c_x^\dagger c_x\Psi_{k''}\rangle \nonumber\\ 
    &= \sum_{x=1}^{|\Gamma|} \sum_{k'} \sum_{j=1}^{N} \alpha_{k'}^* \alpha_{k^{\prime (j)}} \langle\Psi_{k'},c_x^\dagger c_x\Psi_{k^{\prime (j)}}\rangle \\ 
    &= \frac{1}{L} \sum_{j=1}^{N} \sum_{x=1}^{|\Gamma|} \sum_{k'} \alpha_{k'}^* \alpha_{k^{\prime (j)}} e^{ -2i k'_{j} x }  \\ 
    &= \frac{2}{L}\,\RE\left( \sum_{j=1}^N \sum_{x=1}^{|\Gamma|} e^{-2ik_j x} \sum_{k': k_j'>0} \alpha_{k'}^* \alpha_{k^{\prime (j)}}\right)\\
	&= \frac{2}{L} \, \RE\left(\sum_{j=1}^N \frac{e^{-2ik_j}-e^{-2ik_j(|\Gamma|+1)}}{1-e^{-2ik_j}}\sum_{k': k_j'>0} \alpha_{k'}^* \alpha_{k^{\prime (j)}}\right).
\end{align}
\end{subequations}
Next note that with the Cauchy-Schwarz inequality we get
\begin{align}
	\Biggl|\sum_{k': k_j'>0} \alpha_{k'}^* \alpha_{k^{\prime (j)}}\Biggr| \leq \left(\sum_{k'} |\alpha_{k'}|^2 \sum_{k''} |\alpha_{k''}|^2\right)^{1/2} = 1.
\end{align}
Moreover, under the assumption $k_j \in (0, \pi)$, we have that $k_j = \frac{2\pi}{L}\nu_j$ for some $\nu_j \in \{1,\dots,(L-1)/2\}$ if $L$ is odd and $\nu_j\in\{1,\dots,L/2-1\}$ if $L$ is even. Then,
the inequality $|1-e^{-ix}| \geq 2|x|/\pi$ for $x\in [-\pi,\pi]$ implies
\begin{align}
	\left| 1-e^{-2ik_j}\right| = \left|1-e^{-2ik_j+2\pi i} \right| \geq 
    \frac{4}{\pi} \min\{k_j, \left| \pi - k_j\right| \} = \frac{8}{L} \min\left\{\nu_j, \left| \frac{L}{2} - \nu_j\right| \right\} .
\end{align}
We get
\begin{subequations}
\begin{align}
	\Biggl|	\sum_{x=1}^{|\Gamma|}\sum_{k'\neq k''} \alpha_{k'}^* \alpha_{k''} \langle\Psi_{k'},c_x^\dagger c_x\Psi_{k''}\rangle \Biggr| &\leq \frac{2}{L} \sum_{j=1}^N \frac{2}{|1-e^{-2ik_j}|}\\
        &\leq \frac{1}{2} \sum_{j=1}^{N} \max\left\{ \frac{1}{\nu_j}, \left| \frac{L}{2} - \nu_j\right|^{-1} \right\} \\ 
	&\leq \sum_{j=1}^{N} \frac{1}{j} \\ 
	&\leq \ln N + 1,
\end{align}
\end{subequations}
where we use the fact that $\nu_j$'s are pairwise distinct as $k \in \tilde{\mathcal{K}}^{N}$ in deriving the third inequality.

Altogether we therefore obtain
\begin{align}
	\Biggl|\langle\phi,N_\Gamma\phi\rangle  - N\frac{|\Gamma|}{|\Lambda|} \Biggr| \leq \ln N+1.
\end{align}
If $k_{j_0}=0$ or $k_{j_0} = \pi$ for one $j_0$, the computation is basically the same; the only difference is that this index does not appear in the sum over $j$ (which therefore consists only of $N-1$ terms). Moreover, if there are $k_l,k_m$ such that $k_l=-k_m$, then again this only leads to less terms in the sums over $j$. The upper bound $\ln N+1$ thus remains valid also in these cases.

\subsection{Proof of Proposition~\ref{prop: var NGamma Psik}}

We start by computing $\langle\Psi_k,N_\Gamma^2\Psi_k\rangle$. To this end we first note that
\begin{align}
	\langle\Psi_k,N_\Gamma^2\Psi_k\rangle = \sum_{x,y\in \Gamma} \langle\Psi_k,c_x^\dagger c_x c_y^\dagger c_y\Psi_k\rangle. 
\end{align}
If $x=y\in\Gamma$ then
\begin{align}
	\langle\Psi_k,c_x^\dagger c_x c_x^\dagger c_x\Psi_k\rangle = \langle\Psi_k, c_x^\dagger c_x\Psi_k\rangle = \frac{N}{L^d},
\end{align}
see \eqref{eq: exp cx^t cx}. Now suppose that $x\neq y$. Then we find with the help of Lemma~\ref{lem: comb formula} that
\begin{subequations}
	\begin{align}
		\langle&\Psi_k, \,c_x^\dagger c_xc_y^\dagger c_y\Psi_k\rangle \nonumber\\ 
        &= \langle\Phi_{\textup{vac}}, a_{k_N} \dots a_{k_1} c_x^\dagger c_x c_y^\dagger c_y a_{k_1}^\dagger \dots a_{k_N}^\dagger \Phi_{\textup{vac}}\rangle\\
		&= \frac{1}{L^{Nd/2}} \sum_{x_1,\dots,x_N\in \Lambda} e^{-ik_1\cdot x_1} \dots e^{-ik_N \cdot x_N} \langle\Phi_{\textup{vac}}, c_{x_N} \dots c_{x_1} c_x^\dagger c_x c_y^\dagger c_y a_{k_1}^\dagger \dots a_{k_N}^\dagger\Phi_{\textup{vac}}\rangle\\
		&= \frac{1}{L^{Nd/2}} \sum_{x_1,\dots,x_N \in \Lambda} e^{-ik_1\cdot x_1} \dots e^{-ik_N \cdot x_N} \chi_{\{x,y \in \{x_1,\dots,x_N\}  \}} \langle\Phi_{\textup{vac}},c_{x_N} \dots c_{x_1} a_{k_1}^\dagger \dots a_{k_N}^\dagger|\Phi_{\textup{vac}}\rangle\\
		&= \frac{1}{L^{Nd}} \sum_{x_1,\dots x_N \in \Lambda} e^{-ik_1\cdot x_1} \dots e^{-ik_N\cdot x_N} \chi_{\{x,y \in \{x_1,\dots,x_N\}\}} \sum_{\sigma \in \mathcal{S}_N} \mbox{sgn}(\sigma) \prod_{j=1}^N e^{i k_j \cdot x_{\sigma(j)}}\\
		&= \frac{1}{L^{Nd}} \sum_{\substack{l,m=1\\ l\neq m}}^N e^{-ik_l\cdot x} e^{-ik_m \cdot y} \sum_{\substack{x_1,\dots,x_N\in\Lambda\\ x_l=x, x_m=y}} \left(\prod_{\substack{n=1\\ n\neq l,m}}^N e^{-ik_n \cdot x_n}\right) \sum_{\sigma\in\mathcal{S}_N} \mbox{sgn}(\sigma) \prod_{j=1}^N e^{ik_{\sigma^{-1}(j)}\cdot x_j}\\
		&= \frac{1}{L^{Nd}} \sum_{\sigma\in\mathcal{S}_N} \mbox{sgn}(\sigma) \sum_{\substack{l,m=1\\ l\neq m}}^N e^{i(k_{\sigma^{-1}(l)}-k_l)\cdot x} e^{i(k_{\sigma^{-1}(m)}-k_m)\cdot y} \prod_{\substack{n=1\\ n\neq l,m}}^N \left(\sum_{x_n\in\Lambda} e^{i(k_{\sigma^{-1}(n)}-k_n)\cdot x_n}\right).\label{eq: Psik cx^t cx cy^t cy}
	\end{align}
\end{subequations}
Because $\sum_{y \in \Lambda} e^{i(k'-k)\cdot y}= L^{d} \delta_{k,k'} $, see \eqref{eq: sum exp higherdim}, we only get contributions from the permutations $\sigma=\mbox{id}$ and transpositions $\tau_{pq}$ with $p,q\in\{1,\dots N\}$. From $\sigma=\mbox{id}$ we get the contribution
\begin{align}
	\frac{1}{L^{Nd}} N (N-1) L^{(N-2)d} = \frac{N(N-1)}{L^{2d}}
\end{align}
and any transposition $\tau_{pq}$ contributes the term
\begin{multline}
	-\frac{1}{L^{Nd}} \left(e^{i(k_p-k_q)\cdot x} e^{i(k_q-k_p)\cdot y} + e^{i(k_q-k_p)\cdot x} e^{i(k_p-k_q)\cdot y}\right) L^{(N-2)d} \\
 = - \frac{2}{L^{2d}} \RE\left(e^{i(k_p-k_q)\cdot x} e^{i(k_q-k_p)\cdot y}\right).
\end{multline}
Therefore the overall contribution of transpositions becomes
\begin{align}
	-\frac{2}{L^{2d}} \sum_{\substack{p,q=1\\p<q}}^N \RE\left(e^{i(k_p-k_q)\cdot x} e^{i(k_q-k_p)\cdot y}\right).
\end{align}
Altogether we obtain
\begin{align}
\langle\Psi_k,c_x^\dagger c_x c_y^\dagger c_y \Psi_k\rangle = \frac{N(N-1)}{L^{2d}}-\frac{2}{L^{2d}} \sum_{\substack{p,q=1\\p<q}}^N \RE\left(e^{i(k_p-k_q)\cdot x} e^{i(k_q-k_p)\cdot y}\right).
\end{align}
Summing over $x,y\in \Gamma$ we arrive at
\begin{subequations}
\begin{align}
   \langle&\Psi_k, N_\Gamma^2\Psi_k\rangle \nonumber\\
   &= \frac{N}{L^d}|\Gamma| + \frac{N(N-1)}{L^{2d}} |\Gamma| (|\Gamma|-1) - \frac{2}{L^{2d}} \sum_{\substack{x,y\in \Gamma\\ x\neq y}} \sum_{\substack{p,q=1\\p<q}}^N \RE\left(e^{i(k_p-k_q)\cdot x} e^{i(k_q-k_p)\cdot y}\right)\\
	&= \frac{N}{L^d}|\Gamma| + \frac{N(N-1)}{L^{2d}} |\Gamma| (|\Gamma|-1) - \frac{2}{L^{2d}} \sum_{\substack{p,q=1\\p<q}}^N \left(\sum_{x,y\in \Gamma} \RE\left(e^{i(k_p-k_q)\cdot x} e^{i(k_q-k_p)\cdot y}\right)-|\Gamma| \right)\\
	&= \frac{N}{L^d}|\Gamma| + \frac{N(N-1)}{L^{2d}} |\Gamma| (|\Gamma|-1) + \frac{|\Gamma|N(N-1)}{L^{2d}} - \frac{2}{L^{2d}} \sum_{\substack{p,q=1\\p<q}}^N \left|\sum_{x\in\Gamma} e^{i(k_p-k_q)\cdot x} \right|^2\label{ineq: NGamma^2 same ev 1} \\
	&\leq N\frac{|\Gamma |}{|\Lambda|}  + N(N-1) \frac{|\Gamma|^2}{|\Lambda|^2}.\label{ineq: NGamma^2 same ev}
\end{align}
\end{subequations}
With this and $\langle\Psi_k,N_\Gamma\Psi_k\rangle= N |\Gamma|/|\Lambda|$ we finally get for the variance of $N_\Gamma$ in an eigenstate $\Psi_k$ that
\begin{align}
	\langle\Psi_k,N_\Gamma^2\Psi_k\rangle - \left(\langle\Psi_k,N_\Gamma\Psi_k\rangle\right)^2 \leq N \frac{|\Gamma|}{|\Lambda|}\left(1-\frac{|\Gamma|}{|\Lambda|}\right).
\end{align}

\subsection{Proof of Proposition~\ref{prop: var NGamma arb es}}

As in the proof of Proposition~\ref{prop: exp NGamma es} we first assume that there are no $k_l,k_m$ such that $k_l=-k_m$. In this case, we can without loss of generality assume that $k_j\geq 0$ for all $j$.
As we have already computed $\langle\phi,N_\Gamma\phi\rangle$ in Proposition~\ref{prop: exp NGamma es}, it only remains to compute $\langle\phi,N_\Gamma^2\phi\rangle$. To this end, we express $\phi$ again in the basis of the $\Psi_{k'}$ with $k'_j = \pm k_j$, see \eqref{eq: phi}. Then we get
\begin{subequations}
\begin{align}
	\langle\phi,N_\Gamma^2\phi\rangle &= \sum_{k',k''} \alpha_{k'}^* \alpha_{k''} \langle\Psi_{k'},N_\Gamma^2\Psi_{k''}\rangle\\
	&= \sum_{k'} |\alpha_{k'}|^2 \langle\Psi_{k'},N_\Gamma^2\Psi_{k'}\rangle + \sum_{k'\neq k''} \alpha_{k'}^* \alpha_{k''} \langle\Psi_{k'},N_\Gamma^2\Psi_{k''}\rangle.
\end{align}
\end{subequations}
For the first sum we obtain with the help of \eqref{ineq: NGamma^2 same ev} that
\begin{align}
	\sum_{k'} |\alpha_{k'}|^2 \langle\Psi_{k'},N_\Gamma^2\Psi_{k'}\rangle \leq N\frac{|\Gamma|}{|\Lambda|} + N(N-1) \frac{|\Gamma|^2}{|\Lambda|^2}.
\end{align}
For the second sum first note that
\begin{subequations}
\begin{align}
	\sum_{k'\neq k''} &\alpha_{k'}^* \alpha_{k''} \langle\Psi_{k'},N_\Gamma^2\Psi_{k''}\rangle \nonumber\\
    &\qquad = \sum_{x,y\in \Gamma} \sum_{k'\neq k''} \alpha_{k'}^* \alpha_{k''} \langle\Psi_{k'},c_x^\dagger c_x c_y^\dagger c_y\Psi_{k''}\rangle\\
	&\qquad = \sum_{x\in \Gamma} \sum_{k'\neq k''}\alpha_{k'}^*\alpha_{k''} \langle\Psi_{k'},c_x^\dagger c_x\Psi_{k''}\rangle + \sum_{\substack{x,y\in \Gamma\\ x\neq y}} \sum_{k'\neq k''}\alpha_{k'}^* \alpha_{k''} \langle\Psi_{k'},c_x^\dagger c_x c_y^\dagger c_y\Psi_{k''}\rangle.\label{eq: splitting thm 10}
\end{align}
\end{subequations}
The first sum can be estimated as in the proof of Proposition~\ref{prop: exp NGamma es}, i.e.,
\begin{align}
	\Biggl| \sum_{x\in \Gamma} \sum_{k'\neq k''}\alpha_{k'}^*\alpha_{k''} \langle\Psi_{k'},c_x^\dagger c_x\Psi_{k''}\rangle \Biggr| \leq \ln N+1.
\end{align}
For the second sum in \eqref{eq: splitting thm 10} we start by noting that similarly to \eqref{eq: Psik cx^t cx cy^t cy} we have for $x \neq y$ and $k'\neq k''$ that
\begin{multline}\label{eq: cccc different k}
	\langle\Psi_{k'},c_x^\dagger c_x c_y^\dagger c_y\Psi_{k''}\rangle \\
     = \frac{1}{L^N} \sum_{\sigma\in\mathcal{S}_N} \mbox{sgn}(\sigma) \sum_{\substack{l,m=1\\ l\neq m}}^N e^{i(k''_{\sigma^{-1}(l)}-k_l')x} e^{i(k''_{\sigma^{-1}(m)}-k_m')y} \prod_{\substack{n=1\\ n\neq l,m}}^N \left(\sum_{x_n\in\Lambda} e^{i(k''_{\sigma^{-1}(n)}-k'_n)x_n} \right).
\end{multline} 
From this formula we see that if $k'\neq k''$ then only terms where one or two entries of $k'$ and $k''$ are different give non-vanishing contributions.

We shall separate the second sum in \eqref{eq: splitting thm 10} into two sums where $k'$ and $k''$ differ in one or two components, respectively.
Suppose first that $k'$ and $k''$ differ in only one component.
Let us evaluate $\langle\Psi_{k'},c_x^\dagger c_x c_y^\dagger c_y\Psi_{k''}\rangle$ in the case that $k'_1 =k_1>0$ and $k''_1 = -k_1$. If $\sigma = \mbox{id}$, we get the contribution
\begin{align}\label{eq: contribution id}
	\frac{N-1}{L^2}\left(e^{-2ik_1 x} + e^{-2ik_1 y}\right).
\end{align}
The only other permutations that yield non-vanishing terms are transpositions of the form $\tau_{1p}$ with $p>1$. Altogether, these transpositions give the contribution
\begin{align}
-\frac{1}{L^2}\sum_{p=2}^N \left( e^{i(-k_1-k_p')x} e^{i(k_p'-k_1)y} + e^{i(k_p'-k_1)x} e^{i(-k_1-k_p')y} \right).
\end{align}
One obtains analogous expressions in the case that $k'$ and $k''$ do not differ in the first, but in another component.

Our goal is to estimate
\begin{equation}
\sum_{\substack{x,y\in \Gamma\\ x\neq y}} \sum_{\substack{k', k''  \text{differ}\\ \text{ in 1 component}}}\alpha_{k'}^* \alpha_{k''} \langle\Psi_{k'},c_x^\dagger c_x c_y^\dagger c_y\Psi_{k''}\rangle
\end{equation}
First suppose that $0<k_j<\pi$ for all $j$. 
To facilitate the computation, we first let the sum over $x,y\in\Gamma, x\neq y$  run over all $x,y\in\Gamma$ and later estimate the terms where $x=y$.
Considering $k_j'>0$ and $k_j'<0$ separately we can write 
\begin{multline}
\sum_{\substack{x,y\in \Gamma}} \sum_{\substack{k', k''  \text{differ}\\ \text{ in 1 component}}}\alpha_{k'}^* \alpha_{k''} \langle\Psi_{k'},c_x^\dagger c_x c_y^\dagger c_y\Psi_{k''}\rangle\\
=\sum_{\substack{x,y\in \Gamma}} \sum_{j=1}^N \sum_{k':k_j'>0}
\alpha_{k'}^* \alpha_{k^{\prime (j)}} 2\RE\left(\langle\Psi_{k'},c_x^\dagger c_x c_y^\dagger c_y\Psi_{k^{\prime (j)}}\rangle\right),
\end{multline}
where $k^{\prime (j)}$ denotes the vector obtained from $k'$ by flipping the sign of its $j$th component.
Using the expressions for $\langle\Psi_{k'},c_x^\dagger c_x c_y^\dagger c_y\Psi_{k''}\rangle $ derived above, this equals
\begin{align}
	&\frac{2}{L^2} \RE\Biggl(\sum_{x,y=1}^{|\Gamma|}\sum_{j=1}^N\bigg[ (N-1)\left(e^{-2ik_jx} + e^{-2ik_jy}\right)\sum_{\substack{k':k'_j>0}} 
    \alpha_{k'}^* \alpha_{k^{\prime (j)}}\label{eq: one comp diff}\\
	&-\sum_{\substack{p=1\\p\neq j}}^N \big(e^{i(-k_j-k_p)x} e^{i(k_p-k_j)y} + e^{i(k_p-k_j)x} e^{i(-k_j-k_p)y}\big)\sum_{\substack{k':k'_j,k'_p>0}} 
    \alpha_{k'}^* \alpha_{k^{\prime (j)}} \nonumber \\
	&\left.-\sum_{\substack{p=1\\p\neq j}}^N \bigl(e^{i(-k_j+k_p)x} e^{i(-k_p-k_j)y} + e^{i(-k_p-k_j)x} e^{i(-k_j+k_p)y}\bigr) \sum_{k':k'_j>0,k'_p<0} 
    \alpha_{k'}^* \alpha_{k^{\prime (j)}}\bigg]\Biggr) 
    \right. \nonumber 
\end{align}
Carrying out the summations over $x$ and $y$ we arrive at
\begin{align}
	& \frac{2}{L^2} \RE\Biggl(\sum_{j=1}^N\biggl[ 2|\Gamma|(N-1) \frac{e^{-2ik_j}-e^{-2ik_j(|\Gamma|+1)}}{1-e^{-2ik_j}} \sum_{\substack{k':k'_j>0}} 
    \alpha_{k'}^* \alpha_{k^{\prime (j)}}\nonumber\\
	&\quad - 2 \sum_{\substack{p=1\\p\neq j}}^N \biggl(\frac{e^{-i(k_j+k_p)} - e^{-i(k_j+k_p)(|\Gamma|+1)}}{1-e^{-i(k_j+k_p)}} \frac{e^{i(k_p-k_j)}-e^{i(k_p-k_j)(|\Gamma|+1)}}{1-e^{i(k_p-k_j)}}  \sum_{k':k'_j,k_p'>0} \alpha_{k'}^* \alpha_{k^{\prime (j)}} \biggr) \nonumber \\
	&\quad- 2 \sum_{\substack{p=1\\p\neq j}}^N \biggl(\frac{e^{-i(k_j-k_p)} - e^{-i(k_j-k_p)(|\Gamma|+1)}}{1-e^{-i(k_j-k_p)}} \frac{e^{-i(k_p+k_j)}-e^{-i(k_p+k_j)(|\Gamma|+1)}}{1-e^{-i(k_p+k_j)}} \sum_{k':k'_j>0,k'_p<0} \alpha_{k'}^* \alpha_{k^{\prime (j)}}\biggr)\biggr]\Biggr) 
    \label{eq: one comp diff 2}
\end{align}
Taking the absolute value, the sums over $k'$ can again be upper bounded by 1. As in the proof of Proposition~\ref{prop: exp NGamma es}, we compute that 
\begin{multline}
	\frac{4}{L^2}\sum_{j=1}^N \Biggl| |\Gamma|(N-1)\left(\frac{e^{-2ik_j}-e^{-2ik_j(|\Gamma|+1)}}{1-e^{-2ik_j}}\right) \Biggr| \leq \frac{4|\Gamma|(N-1)}{L^2} \sum_{j=1}^N \frac{2}{(4j/L)}\\
	= \frac{2|\Gamma|(N-1)}{L}\sum_{j=1}^N \frac{1}{j}
	\leq \frac{2|\Gamma|(N-1)}{L} (\ln N+1).
\end{multline}
If $k_{j_0}= 0$ or $k_{j_0}=\pi$ for some $j_0$ or if there are $k_l,k_m$ such that $k_l=-k_m$, the corresponding term is missing in the sum over $j$ and the upper bound of $N-1$ remains valid.

Next we estimate
\begin{subequations}
\begin{align}
	\frac{4}{L^2} &\sum_{j=1}^N\Biggl|\sum_{\substack{p=1\\p\neq j}}^N \frac{e^{-i(k_j+k_p)} - e^{-i(k_j+k_p)(|\Gamma|+1)}}{1-e^{-i(k_p+k_j)}} \frac{e^{i(k_p-k_j)}-e^{i(k_p-k_j)(|\Gamma|+1)}}{1-e^{i(k_p-k_j)}} \Biggr|\\
	&\leq \frac{16}{L^2} \sum_{j=1}^N\sum_{\substack{p=1\\p\neq j}}^N \frac{1}{|1-e^{-i(k_j+k_p)}| |1-e^{i(k_p-k_j)}|}\\
 &\leq \frac{16}{L^2} \sum_{j=1}^N \left(\sum_{\substack{p=1\\p\neq j}}^N \frac{1}{|1-e^{-i(k_p+k_j)}|^2}\right)^{1/2} \left(\sum_{\substack{p=1\\p\neq j}}^N \frac{1}{|1-e^{i(k_p-k_j)}|^2}\right)^{1/2}\\
 &\leq \frac{16}{L^2} \sum_{j=1}^N \left(\sum_{p=1}^N \frac{L^2}{16p^2}\right)\\
 &\leq 2N,
\end{align}
\end{subequations}
where we used in the last step that $\sum_{n=1}^\infty n^{-2}=\pi^2/6<2$. Note that the last term in \eqref{eq: one comp diff 2} is of a similar form and can be estimated along the same lines.

To get an estimate for the overall contribution of $k'\neq k''$ which differ in one component, we still have to estimate the terms in \eqref{eq: one comp diff} with $x=y$. By a similar computation as before and in the proof of Proposition~\ref{prop: exp NGamma es} we find that an upper bound for the absolute value of these terms is given by
\begin{align}
    \frac{12(N-1)}{L^2} \sum_{j=1}^N \Biggl| \frac{e^{-2ik_j}-e^{-2ik_j(|\Gamma|+1)}}{1-e^{-2ik_j}} \Biggr| \leq \frac{6(N-1)}{L}(\ln N+1).
\end{align}
Again, if there is a $k_{j_0}$ that is equal to 0 or $\pi$ or if there are $k_l,k_m$ such that $k_l=-k_m$, the only modification that has to be made is to exclude the terms in the sum over $j$ and therefore the estimates remain valid. 

Altogether we find for the contribution of the terms in \eqref{eq: splitting thm 10} with $k'\neq k''$, where $k'$ and $k''$ differ in exactly one component, the upper bound
\begin{align}
	\frac{2|\Gamma|(N-1)}{L}(\ln N+1)+4N + \frac{6(N-1)}{L}\left(\ln N+1\right).
\end{align}

Next we turn to the terms in \eqref{eq: splitting thm 10} with $k'\neq k''$ which differ in two components. As before, we can assume that $0<k_j<\pi$ for all $j$ as by the same reasoning as in the previous computations, the upper bounds remain valid if this condition is relaxed.

Consider \eqref{eq: cccc different k} in the case $k'_1=k_1>0, k_2'=k_2>0$ and $k_1''=-k_1, k''_2=-k_2$. Then the permutation $\sigma=\mbox{id}$ gives the contribution
\begin{align}
	\frac{1}{L^2}\left(e^{-2ik_1x} e^{-2ik_2y} + e^{-2ik_2x} e^{-2ik_1y}\right)
\end{align}
and from the permutation $\tau_{12}$ we get
\begin{align}
	-\frac{2}{L^2} e^{-i(k_1+k_2)x} e^{-i(k_1+k_2)y}.
\end{align}
If $k'_1=k_1, k_2'=-k_2, k''_1=-k_1, k''_2=k_2$, the permutation $\sigma =\mbox{id}$ gives
\begin{align}
	\frac{1}{L^2} \left(e^{-2ik_1x} e^{2ik_2y}+e^{2ik_2x}e^{-2ik_1y}\right)
\end{align}
and $\tau_{12}$ yields the term
\begin{align}
	-\frac{2}{L^2} e^{i(k_2-k_1)x} e^{i(k_2-k_1)y}.
\end{align}
Summing again first over all $x,y\in\Gamma$ and later estimating the terms with $x=y$ which were added to facilitate the computation we get
\begin{align}
	\frac{2}{L^2}&\RE\Biggl(\sum_{x,y=1}^{|\Gamma|}\sum_{\substack{l,m=1\\l\neq m}}^N\Bigg[\left(e^{-2ik_lx} e^{-2ik_my} + e^{-2ik_mx}e^{-2ik_ly}- 2 e^{-i(k_l+k_m)x} e^{-i(k_l+k_m)y}\right) \sum_{k':k'_l,k'_m>0} \alpha_{k'}^* \alpha_{k^{\prime (lm)}}\nonumber\\
    &+ \left(e^{-2ik_lx} e^{2ik_my} + e^{2ik_mx} e^{-2ik_l y} - 2 e^{i(k_m-k_l)x} e^{i(k_m-k_l)y}\right) \sum_{k':k'_l>0,k'_m<0} \alpha_{k'}^* \alpha_{k^{\prime (lm)}} \Bigg]\Biggr),
\end{align}
where $k^{\prime (lm)}$ denotes the vector obtained from $k'$ by flipping the signs of its $l$th and $m$th components.
Using similar estimates as previously, we can bound the absolute value of this term by
\begin{align}
    2(\ln N+1)^2 + 4N.
\end{align}
Next we note that the terms with $x=y$ in this sum vanish. 

Putting everything together, we finally arrive at
\begin{subequations}
\begin{align}
	\langle\phi,N_\Gamma^2\phi\rangle &\leq N\frac{|\Gamma|}{|\Lambda|} + N(N-1) \frac{|\Gamma|^2}{|\Lambda|^2} + \ln N +1 + \frac{2|\Gamma|(N-1)}{L}(\ln N+1) + 4N \nonumber\\
	&\quad + \frac{6(N-1)}{L}(\ln N+1)  + 2 (\ln N+1)^2 + 4N\\
	&\leq N^2 \frac{|\Gamma|^2}{|\Lambda|^2} + 2 N \ln N+ 11N + 2(\ln N)^2+ 11 \ln N +9.
\end{align}
\end{subequations}
Because of $(a^2-b^2)= (a-b)(a+b)$ for $a,b\in\mathbb{R}$ it follows from Proposition~\ref{prop: exp NGamma es} that
\begin{align}
	\Biggl|\langle\phi,N_\Gamma\phi\rangle^2 - N^2 \frac{|\Gamma|^2}{|\Lambda|^2} \Biggr| \leq  \left(\ln N+1\right)\left(\ln N +1 + 2N\frac{|\Gamma|}{|\Lambda|}\right)
\end{align}
and therefore 
\begin{align}
	\langle\phi,N_\Gamma\phi\rangle^2 \geq N^2 \frac{|\Gamma|^2}{|\Lambda|^2}- \left(\ln N+1\right)\left(\ln N +1 + 2N\frac{|\Gamma|}{|\Lambda|}\right).
\end{align}
With this we finally obtain
\begin{align}
	\langle\phi,N_\Gamma^2\phi\rangle - \langle\phi,N_\Gamma\phi\rangle^2 \leq 4 N \ln N + 13 N + 3(\ln N)^2 + 13 \ln N + 10\,.
\end{align}
For $N\geq 46$ it holds that $13 N + 3(\ln N)^2 + 13 \ln N + 10<4N\ln N$.
This finishes the proof.

\subsection{Proof of Proposition~\ref{prop: ETH higherdim}}

	The proof of a similar result for $d=1$ given by Tasaki in \cite{T24} can be adapted to our situation with only very small modifications. Note that the proof in the present situation in the case that $d=1$ was already given by Tasaki in an earlier arXiv version of his paper, however, the model was slightly changed in later versions.
 
	Similarly as in the proof of Lemma~4 in \cite{T24} we first show that
	\begin{align}
		\langle\Psi_k, e^{\lambda N_{\Gamma}} \Psi_k\rangle \leq \left(\mu e^\lambda + (1-\mu)\right)^N 
	\end{align}
for any $\lambda \in (0,1]$. To this end note that
\begin{align}
	e^{\lambda N_{\Gamma}/2} |\Psi_k\rangle = b_{k_1}^\dagger \dots b_{k_N}^\dagger |\Phi_{\textup{vac}}\rangle
\end{align}
with
\begin{align}
	b_k^\dagger \coloneqq  \frac{1}{L^{d/2}} \left(e^{\lambda/2} \sum_{x\in\Gamma} e^{ik\cdot x} c_x^\dagger + \sum_{x\in\Lambda\backslash \Gamma}e^{ik\cdot x} c_x^\dagger\right).
\end{align}
We then obtain that
\begin{align}
	\langle \Psi_k, e^{\lambda N_\Gamma} \Psi_k\rangle = \langle\Phi_{\textup{vac}},  b_{k_N} \dots b_{k_1} b_{k_1}^\dagger \dots b_{k_N}^\dagger \Phi_{\textup{vac}}\rangle &\leq \prod_{j=1}^N \|b_{k_j} b_{k_j}^\dagger\|.
\end{align}
For any $j$, the operator $b_{k_j} b_{k_j}^\dagger$ is self-adjoint and 
\begin{multline}
	\left\{b_{k_j}, b_{k_j}^\dagger\right\} = \frac{1}{L^d}\left(e^{\lambda} \sum_{x,y\in\Gamma} e^{ik_j\cdot (x-y)} \{c_y, c_x^\dagger\} + e^{\lambda/2} \sum_{\substack{x\in \Gamma\\ y\in \Lambda\backslash \Gamma}} e^{ik_j\cdot (x-y)} \{c_y, c_x^\dagger\}\right.\\
	\left.+ e^{\lambda/2} \sum_{\substack{x \in \Lambda\backslash\Gamma\\ y\in \Gamma}} e^{ik_j\cdot (x-y)} \{c_y,c_x^\dagger\} + \sum_{x,y\in \Lambda\backslash\Gamma}e^{ik_j\cdot(x-y)} \{c_y,c_x^\dagger\}\right)\\
	= \frac{1}{L^d} \left(e^\lambda |\Gamma| + |\Lambda|-|\Gamma|\right)
	= \mu e^{\lambda} + (1-\mu).
\end{multline}
Next note that this implies
\begin{align}
	(b_{k_j}b_{k_j}^\dagger)^2 = \left(\mu e^\lambda + (1-\mu)\right) b_{k_j} b_{k_j}^\dagger - b_{k_j} b_{k_j} b_{k_j}^\dagger b_{k_j}^\dagger = (\mu e^\lambda + (1-\mu)) b_{k_j} b_{k_j}^\dagger.\label{eq: bkj square}
\end{align}
The last step follows from $b_{k_j} b_{k_j}=0$, which can be seen as follows: From the definition of the $b_{k_j}$ we immediately obtain
\begin{multline}
	b_{k_j} b_{k_j} = \frac{1}{L^d}\left(e^{\lambda} \sum_{x,y\in\Gamma} e^{-ik_j\cdot (x+y)} c_x c_y + e^{\lambda/2} \sum_{\substack{x\in \Gamma\\ y\in \Lambda\backslash \Gamma}} e^{-ik_j\cdot (x+y)} c_x c_y\right. \\
	\left.+ e^{\lambda/2} \sum_{\substack{x \in \Lambda\backslash\Gamma\\ y\in \Gamma}} e^{-ik_j\cdot (x+y)} c_x c_y + \sum_{x,y\in \Lambda\backslash\Gamma}e^{-ik_j\cdot(x+y)} c_x c_y\right).
\end{multline}
If $x=y$, then $c_x c_y=0$. In the first and fourth sum, for every term $c_x c_y$ with $x\neq y$ also the term $c_yc_x$ occurs and because of $c_x c_y=-c_y c_x$ and the same prefactors, the terms cancel. Therefore, the first and fourth sum vanishes. By a similar argumentation, the second sum is equal to the $(-1)$ times the third sum, i.e., they cancel, and altogether we obtain that $b_{k_j} b_{k_j}=0$.

It follows from \eqref{eq: bkj square} that the eigenvalues of the self-adjoint operator $b_{k_j} b_{k_j}^\dagger$ are 0 and $\mu e^{\lambda} + (1-\mu)$: If $\alpha \in \mathbb{R}$ is an eigenvalue of $b_{k_j}b_{k_j}^\dagger$ with eigenfunction $\phi$, then $b_{k_j} b_{k_j}^\dagger |\phi\rangle = \alpha |\phi\rangle$ and therefore $(b_{k_j} b_{k_j}^\dagger)^2|\phi\rangle = \alpha b_{k_j} b_{k_j}^\dagger|\phi\rangle$, but at the same time $(b_{k_j} b_{k_j}^\dagger)^2|\phi\rangle = (\mu e^\lambda+(1-\mu)) b_{k_j} b_{k_j}^\dagger|\phi\rangle$, i.e., either $\alpha = \mu e^\lambda + (1-\mu)$ or $b_{k_j}b_{k_j}^\dagger|\phi\rangle = 0$ which implies $\alpha=0$ as $\phi\neq 0$. Therefore we conclude that $\|b_{k_j} b_{k_j}^\dagger\| = \mu e^\lambda + (1-\mu)$ and hence
\begin{align}\label{eq: exp-ng}
	\langle\Psi_k, e^{\lambda N_\Gamma} \Psi_k\rangle \leq \left(\mu e^\lambda + (1-\mu)\right)^N.
\end{align}
The rest of the proof of \eqref{eq: ETH higherdim} is exactly the same as in the first arXiv version of \cite{T24}, but we include it here for the convenience of the reader. 
First note that the projection onto the non-equilibrium space can be bounded in terms of characteristic functions as
\be \label{eq: Pneq-characteristic-fn}
\bigl\| P_{\noneq,\eta}\Psi_k\bigr\|^2 \leq \langle \Psi_k, \mathbf{1}_{[N (\mu+\eta),N]}(N_\Gamma) \Psi_k\rangle +  \langle \Psi_k, \mathbf{1}_{[0,N (\mu-\eta)]}(N_\Gamma) \Psi_k\rangle.
\ee
In the first term, the characteristic function is bounded above by $e^{\lambda (N_\Gamma/N-\mu-\eta)N}$  for all $\lambda\geq 0$.
Using inequality \eqref{eq: exp-ng} we obtain
\be
\langle \Psi_k, \mathbf{1}_{[N (\mu+\eta),N]}(N_\Gamma) \Psi_k\rangle \leq \left[(\mu e^{\lambda(1-\mu)}+(1-\mu)e^{-\lambda \mu})e^{-\lambda \eta}\right]^N
\ee
Using the power series representation for the exponential function one finds that for all $0<\mu<1, |\lambda|<1$
\be
\mu e^{\lambda(1-\mu)}+(1-\mu)e^{-\lambda \mu} \leq 1+\frac{\mu(1-\mu)}{2}\lambda^2+\frac{\mu(1-\mu)}{2}\lambda^2 \sum_{l=3}^\infty \frac{1}{l!} \leq 1+\frac{3}{4}\mu(1-\mu)\lambda^2\leq e^{\frac{3}{4}\mu(1-\mu)\lambda^2}.
\ee
Choosing $\lambda=2\eta/(3 \mu(1-\mu))$, which is smaller than $1$ by assumption, results in the bound
\be\label{eq: ng-bound}
\langle \Psi_k, \mathbf{1}_{[N (\mu+\eta),N]}(N_\Gamma) \Psi_k\rangle 
\leq e^{-\frac{\eta^2}{3\mu(1-\mu)}N}.
\ee
For the second term in \eqref{eq: Pneq-characteristic-fn}, note that it equals $\langle \Psi_k, \mathbf{1}_{[N(1-\mu+\eta),N]}(N_{\Lambda\setminus\Gamma}) \Psi_k\rangle$.
Now we can apply the bound \eqref{eq: ng-bound} with $\Gamma$ replaced by $\Lambda\setminus \Gamma$ and $\mu$ by $(1-\mu)$ and conclude that
\be
\bigl\| P_{\noneq,\eta}\Psi_k\bigr\|^2 < 2 e^{-\frac{\eta^2}{3\mu(1-\mu)}N}.
\ee

The last statement of Proposition~\ref{prop: ETH higherdim}, that $D_E\geq 2^{Nd}$ if $N<L/2d$, can be verified as follows. Choose $Nd$ distinct positive integers less than $L/2$, multiply them by $2\pi/L$, and call them, in any order, $k_{ia}$ with $i=1,\ldots,N$ and $a=1,\ldots,d$. Write $k_i=(k_{i1},\ldots, k_{id})$ and $\tilde{k}_i=k_{\pi(i)}$ for the permuted version such that $\tilde{k}=(\tilde{k}_1,\ldots,\tilde{k}_N)\in \tilde{\mathcal{K}}^N$. Thus, $\Psi_{\tilde k}$ as in \eqref{eq: energy eigenstate highdim} 
is an eigenvector of $\HfF$, and the corresponding eigenvalue is $-2\sum_{i=1}^N \sum_{a=1}^d \cos k_{ia}$. 
Now consider the $2^{Nd}$ $Nd$-vectors $k'$ with $k_{ia}'=\pm k_{ia}$; none of them is a permutation of any other, so suitable permutations yield $2^{Nd}$ distinct elements $\tilde{k}'$ of $\tilde{\mathcal{K}}^N$. The corresponding $\Psi_{\tilde{k}'}$ have the same eigenvalue, so the eigenspace must at least have dimension $2^{Nd}$.

\section{Conclusions}
\label{sec:conclusions}

If a Hamiltonian $H$ satisfies the ETH in the form that every eigenvector is in MATE (which we show in Theorem~\ref{thm2} is the case for the free Fermi gas on a 1d lattice), then every initial pure state $\psi_0$ will thermalize in the sense of MATE. 
For a highly degenerate Hamiltonian $H_0$ for which at least \emph{one} eigenbasis lies in MATE, in general not all eigenbases need to lie in MATE, and hence not all initial states thermalize.
Nevertheless, as Proposition~\ref{prop: most eigenbases H0} shows, \textit{most} eigenbases of such a $H_0$ do lie in MATE. 
This leads to Theorem~\ref{thm1 new}, which shows that adding an arbitrarily small generic perturbation to $H_0$ will with high probability yield a non-degenerate Hamiltonian $H$ whose eigenbasis is indeed in MATE, and thus $H$ exhibits thermalization for all initial states.
As a concrete example, these general considerations apply to the free Fermi gas on a $d$-dimensional lattice with $d>1$.

\bigskip
\bigskip

\noindent{\bf Acknowledgments.} We thank Hal Tasaki and Peter Reimann for very valuable feedback on the first version of this paper, Herbert Spohn for additional references,  and Hannah Markwig and Thomas Markwig for help with Footnote~\ref{fn:alg}. 
This work was supported by the Deutsche Forschungsgemeinschaft (DFG, German Research Foundation) -- TRR 352 -- Project-ID 470903074.
C.V.\ acknowledges financial support by the German Academic Scholarship Foundation. Moreover, the work of C.V.\ has been partially supported by the ERC Starting Grant ``FermiMath", grant agreement nr. 101040991, funded by the European Union. Views and opinions expressed are however those of the authors only and do not necessarily reflect those of the European Union or the European Research Council Executive Agency. Neither the European Union nor the granting authority can be held responsible for them.\\
\\
\noindent{\bf Data Availability Statement.} No data were analyzed in this paper.\\
\\
\noindent{\bf Conflict of Interest Statement.} The authors have no conflicts of interest.

\bibliographystyle{plainurl}
\bibliography{literature.bib}

\end{document}